\documentclass[a4paper,11pt]{article}
\pdfoutput=1 

\usepackage{jheppub} 
\makeatletter
\gdef\@fpheader{}
\makeatother

\usepackage[T1]{fontenc} 

\usepackage{amsthm}
\usepackage[T1]{fontenc} 
\usepackage{dsfont}
\usepackage{amsmath}
\usepackage{bbm}
\usepackage[capitalize]{cleveref}
\usepackage{physics}
\usepackage[normalem]{ulem}
\usepackage{caption}
\usepackage{subcaption}
\usepackage{adjustbox}
\usepackage{comment}
\usepackage{cleveref}
\usepackage[section]{placeins}
\usepackage{mathrsfs}

\newtheorem{theorem}{Theorem}[section]
\newtheorem{assumption}{Assumption}[section]
\newtheorem{claim}{Claim}[section]
\newtheorem{definition}{Definition}[section]
\newtheorem{lemma}{Lemma}[section]

\newtheorem{primitive}[theorem]{Primitive}

\newcommand{\defi}{:=}
\DeclareMathOperator{\supp}{supp}
\DeclareMathOperator{\spn}{span}
\newcommand{\north}{{S_{N}}}
\newcommand{\south}{{S_{S}}}
\newcommand{\east}{{S_{E}}}
\newcommand{\west}{{S_{W}}}
\newcommand{\Csys}{{\mathcal{S}}}
\newcommand{\Cnorth}{{\mathcal{S}_{N}}}
\newcommand{\Csouth}{{\mathcal{S}_{S}}}
\newcommand{\Ceast}{{\mathcal{S}_{E}}}
\newcommand{\Cwest}{{\mathcal{S}_{W}}}
\newcommand{\Id}{\mathds{1}}
\newcommand{\vecspace}{\mathcal{M}}
\newcommand{\vecspaceop}{M}

\newcommand{\Xl}{X_{\mathrm{L}}}
\newcommand{\tXl}{\tilde{X}_{\mathrm{L}}}
\newcommand{\Xb}{X_{\mathrm{blk}}}
\newcommand{\tXb}{\tilde{X}_{\mathrm{blk}}}
\newcommand{\zc}{0_{\mathrm{CFT}}}
\newcommand{\Xc}{X_{\mathrm{CFT}}}
\newcommand{\tXc}{\tilde{X}_{\mathrm{CFT}}}

\newcommand{\Xs}{X_{\mathrm{sim}}}
\newcommand{\tXs}{\tilde{X}_{\mathrm{sim}}}
\newcommand{\opnorm}[1]{||#1||}
\newcommand{\dnorm}[1]{||#1||_\diamond}

\newcommand{\angcutoff}{\theta}

\title{Holography as a resource for non-local quantum computation}

\author[a]{Kfir Dolev}
\author[a]{Sam Cree}

\affiliation[a]{Stanford Institute for Theoretical Physics, Stanford University, 382 Via Pueblo Mall, Stanford, CA 94305-4060, U.S.A.}

\emailAdd{Dolev@stanford.edu}

\abstract{
If two parties share sufficient entanglement, they are able to implement any channel on a shared bipartite state via non-local quantum computation -- a protocol consisting of local operations and a single simultaneous round of quantum communication.
Such a protocol can occur in the AdS/CFT correspondence, with the two parties represented by regions of the CFT, and the holographic state serving as a resource to provide the necessary entanglement.
This boundary non-local computation is dual to the local implementation of a channel in the bulk AdS theory.
Previous work on this phenomenon was obstructed by the divergent entanglement between adjacent CFT regions, and tried to circumvent this issue by assuming that certain regions are irrelevant.
However, the absence of these regions introduces violent phenomena that prevent the CFT from implementing the protocol.
Instead, we resolve the issue of divergent entanglement by using a finite-memory quantum simulation of the CFT.
We show that any finite-memory quantum system on a circular lattice yields a protocol for non-local quantum computation.
In the case of a quantum simulation of a holographic CFT, we carefully show that this protocol implements the channel performed by the local bulk dynamics.
Under plausible physical assumptions about quantum computation in the bulk, our results imply that non-local quantum computation can be performed for any polynomially complex unitary with a polynomial amount of entanglement.
Finally, we provide a concrete example of a holographic code whose bulk dynamics correspond to a Clifford gate, and use our results to show that this corresponds to a non-local quantum computation protocol for this gate.
}

\begin{document} 
\maketitle
\flushbottom

\section{Introduction }

In position-based cryptography \cite{Buhrman_2014,PBQC,practical-PBQC,tagging,tagging2}, individuals use their spacetime position as cryptographic credentials.
In the simplest task, position verification, a prover must prove to a verifier that they are in a specific spacetime region.
They do so by performing a local computation on signals sent by the verifier, and returning the outputs back to the verifier.
The protocol is designed so that only someone in the authorized location is capable of locally performing the computation without violating causality.
This may be used, for example, to ensure that only someone inside a trusted facility is able to read a message.
However, two dishonest provers can collaborate to non-locally simulate this local computation using \emph{non-local quantum computation} (NLQC) -- a quantum task in which two parties implement a joint channel on the systems they hold when limited to one round of communication \cite{tagging,tagging2,T-gate-protocol,QuantumPseudoTelepathy,PBQC,practical-PBQC,new-ideas,beating-classical,banach-spaces,GH,Logspace-routing,loss-tolerant,conundrum,insecurity,Code-routing,Unruh2014QuantumPV,homomorphic-encryption,Chitambar,Broadbent_2016,Yu_2012,Chakraborty_2015} -- provided they share enough resources to do so.

The question thus remains to characterize exactly how resource-intensive NLQC is.
Only sufficiently simple computations are relevant for position-based cryptography, since a computation that is too complex cannot be performed in time even by an honest prover without jeopardizing the causality constraints.
Thus if all low-complexity computations can be efficiently performed non-locally -- i.e.\ with only a polynomial amount of resources --
then practical, secure position verification is impossible.

There exist some partial results about characterizing the resource  requirements for NLQC, which focus on quantifying the number of Bell pairs required.
The tightest known resource requirement upper bound for general unitaries was derived in Ref.~\cite{beigi-koenig} by giving an explicit general purpose protocol.
This protocol consumes a number of Bell pairs exponential in the total number of qubits on which the unitary acts.
More efficient protocols can sometimes be found by exploiting the structure of the unitary \cite{T-gate-protocol,Small-lightcones} or restricting it to a particular class \cite{Code-routing,Garden-Hose}.
Linear lower bounds are known for specific tasks \cite{monogamy,Alex-complexity,Bluhm_2022}. For some unitaries, a loglog lower bound in terms of complexity is given in Ref.~\cite{Alex-complexity}.

The anti-de Sitter/conformal field theory (AdS/CFT) correspondence \cite{Maldacena1997,Witten1998} is a family of holographic dualities \cite{Hooft1993,Susskind1995} that relate a bulk theory of quantum gravity in asymptotically AdS spacetime to a CFT living on its boundary. 
A notable pattern in these dualities is that information-theoretic quantities in the boundary are related to geometric quantities in the bulk \cite{Ryu_2006,Hubeny_2007,Alex-connect-wedge-theorem,Reflected-Entropy,ER=EPR}.
It was recently argued that the boundary dynamics of certain holographic systems can be interpreted as executing an NLQC protocol that implements the bulk dynamics \cite{Alex-tasks}.
This observation has yielded a number of results and conjectures about both AdS/CFT and NLQC \cite{Alex-tasks, Alex-connect-wedge-theorem,Alex-complexity}, such as a method of placing constraints on bulk dynamics, and a tension between the existence of universal quantum computers in holography and the possibility of secure position-based cryptography. 
The tension arises from the claim that if such a computer could be placed in the bulk, the boundary could perform any unitary non-locally with an amount of entanglement scaling at most polynomially with the complexity.
This would dramatically improve upon the general purpose protocol of Ref.~\cite{beigi-koenig} for simple unitaries.
If this is demonstrated rigorously, then all simple unitaries could be efficiently implemented non-locally, and thus secure position verification would be impossible.
However, it remains to establish a precise connection between this behavior and the task of NLQC, as previous attempts have been non-rigorous, and relied on significant assumptions whose validity we question in this work.

Here we provide a more careful and detailed demonstration of this connection without relying on those assumptions.
We formally establish the connection between holography and NLQC by carefully showing that it is possible to extract a protocol from a simulation of the boundary CFT that implements the channel associated with the local bulk dynamics.
The simulation only needs to accurately capture simple correlation functions of certain operators, preserve locality of operators, and satisfy an approximate light-cone, in ways that we make precise.
Our protocol depends only on the initial CFT simulation state and its Hamiltonian, ensuring that it captures the particular mechanism that AdS/CFT uses to accomplish the task; this is in contrast to the suggestion proposed in Ref.~\cite{Alex-connect-wedge-theorem}, which we show in \cref{sec: relation-to-previous-work} requires use of operations not performed by the CFT itself.
It also explicitly uses finite-memory\footnote{i.e.\ a finite number of qubits} quantum systems, rather than the infinite-dimensional systems associated with regions of a field theory.

We start by introducing a general-purpose ``many-body'' NLQC protocol for any finite memory $(1+1)$-dimensional quantum system living on a circular lattice.
It uses an initial state of the system as the resource, and implements some computation using the local dynamics of the system.
Such systems mimic the causal structure of a $(1+1)$-dimensional CFT since the locality of the Hamiltonian gives rise to an emergent light-cone described by the Lieb-Robinson velocity \cite{Lieb-Robinson}. 
When this protocol is applied to a quantum system simulating the relevant features of a holographic CFT, the computation that is implemented is exactly the one corresponding to the bulk dynamics of interest (see \cref{fig:intro-dream-within-a-dream}). 

It is not clear how to simulate a CFT with a $(1+1)$-dimensional system suitable for NLQC protocol extraction, such that the underlying computation is the same.
A perfect simulation of a field theory with a finite-dimensional system is impossible, because subregions of quantum field theories have divergent entropy.
Fortunately, any CFT is renormalizable, meaning that only a finite subset of its degrees of freedom are ever relevant to a particular phenomena. Thus in order to preserve the NLQC performed by a holographic CFT, a simulation of it only needs to capture the subset of degrees of freedom relevant to the bulk computation, and their relevant dynamics.
Thus we look for a minimal set of CFT features the simulation must reproduce in order for the associated protocol to implement the bulk dynamics.
We find these boil down to 1.\ preserving the locality structure of the CFT, 2.\ reproducing certain low-order correlation functions, and either 3a.\ an approximate light cone is satisfied even for operators without counterparts in the CFT (e.g.\ due to a Lieb-Robinson velocity \cite{Lieb-Robinson}), or 3b.\ Cauchy slices\footnote{A Cauchy slice is a surface which every time like curve without end points crosses exactly once.} can be locally deformed. 
The correlation functions are generally captured only inexactly, up to some error quantified by a dimensionless parameter $\delta$.
The approximate light-cone means that the commutator of two operators is bounded proportional to $a \exp(-b(d-t))$, where $d$ and $t$ are the space and time separations of the operators, and $a$ and $b$ are error parameters.
As the resolution of the simulation is increased, $b$ should diverge to infinity so that the light-cone becomes exact, and $a$ should not increase so quickly that the bound becomes trivial in that limit.

We argue for the existence of such a simulation of a holographic CFT by looking at known examples of other CFT simulations.
In particular, the simulation method of Ref.~\cite{Osborne-Stottmeister} makes use of a rigorous definition of a CFT using a technique known as operator-algebraic renormalization \cite{Morinelli_2021,Brothier_2019,Stottmeister-Vincenzo-etal}, and captures at least features 1, 2 and 3b.
On the other hand, Ref.~\cite{Hamma_2009} considers a lattice system whose low energy limit reproduces a continuum $\mathrm U(1)$ gauge theory in a way that captures features $1,2$, and 3a sufficiently well. 
No rigorous results are known about simulating strongly coupled QFTs (such as holographic CFTs) due to technical challenges, but it is widely believed that similar techniques should apply to them.

Our main result is the following.
Let $\mathcal{U}_{L}$ be a channel representing the bulk dynamics that we construct an NLQC protocol for.
Alice and Bob will use an encoding channel $\mathcal{E}_0$ to input their respective systems into the holographic simulation.
We show that they can then apply an approximation $\mathcal{V}_{sim}$ of the local time evolution in the delocalized form required for NLQC, i.e.\ using local operations and one round of communication.
Finally, they use a recovery channel $\mathcal{R}_\tau$ to decode their output systems from the simulation.
Aside from the simulation's error parameters $\delta$, $a$ and $b$, there will also be error arising from the inexactness of the semiclassical description of the bulk.
The emergence of the bulk is only exact in the limit of weak gravity, that is, as Newton's constant $G_N\to 0$.
We show that the error in the application of $\mathcal{U}_L$ goes as
\begin{align*}
    \dnorm{\mathcal{R}_\tau\circ\mathcal{V}_{sim}\circ\mathcal{E}_0-\mathcal{U}_{L}}\leq c_{CFT}\sqrt{G_N}+c_{sim}\sqrt{\delta}+c_{spread} a \exp(-b \Delta \tau),
\end{align*}
where $c_{CFT}$, $c_{spread}$, $c_{sim}$ and $\Delta \tau$ are $O(1)$ positive parameters, and $G_N$ is the Newton's constant of the holographic theory.
All these errors can be made arbitrarily small by increasing the resolution scale of the simulation, and decreasing Newton's constant; however these will both be associated with increases in the amount of entanglement required.
Loosely speaking, increasing the resolution scale means more physical sites, and decreasing $G_N$ means more qudits per site.
This gives some intuition for why the necessary entanglement increases, as the Hilbert space is simply larger.
\begin{figure}
\centering
\includegraphics[width=0.7\textwidth]{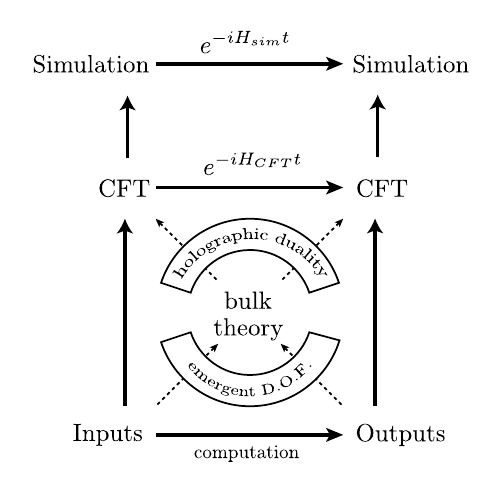}
\caption{Relationship between the various systems we consider. Systems connected by solid lines are related by CPTP maps.  Systems connected by dotted lines are related by equality of correlation functions, which is necessary because the bulk does not have a clear Hilbert space structure. The input/output as well as the simulation systems have finite memory and are thus easily manipulated. The computation on the input/output systems describes the bulk dynamics of interest, while the simulation describes the degrees of freedom in the CFT that encode these dynamics non-locally.}
\label{fig:intro-dream-within-a-dream}
\end{figure}

The entanglement associated with a holographic resource state depends on the parameters of the CFT theory, and the geometry of the bulk region in which the local computation is performed.
These parameters can be tuned to allow increasingly complex unitaries in the bulk, at the expense of an increase in entanglement.
We present a non-rigorous physical argument for an assumption about the ability to perform quantum computations in the bulk of AdS.
Under this assumption, as well as the assumption that a suitable simulation method for holographic CFTs exists, the entanglement cost of performing NLQC scales polynomially with the time and space complexity of the unitary.
Such a result would imply that position verification is fundamentally insecure, as any polynomial-time verification protocol can be efficiently spoofed by a pair of dishonest collaborators (i.e.\ with polynomial entanglement).

As an explicit example of an NLQC protocol extracted from holography, we describe a holographic code with bulk dynamics that implements an arbitrary $n$-qubit Clifford gate.
The many-body NLQC protocol then shows how to use the code as a resource to perform that Clifford gate non-locally.
The dynamics are achieved by a combination of discrete translations such as in Ref.~\cite{Osborne-Stiegemann} and the fact that the holographic code introduced in Ref.~\cite{CSS-holographic-codes} admits transversal Clifford gates.
Interestingly, the code space in which the incoming and outgoing systems live must be different for the construction to work.

Previous attempts at deriving an upper bound on entanglement for holographic NLQC protocols \cite{Alex-complexity} have assumed that entanglement in certain subregions of the CFT do not contribute to the protocol.
However, we argue that the existence of these regions prevents violent phenomena in the bulk and thus they are essential for a protocol to arise from the CFT dynamics. 

The paper is organized as follows.
In \cref{sec: background} we give background on NLQC, AdS/CFT and the connection between them.
In \cref{sec: many-body-nlqc} we derive a ``many-body'' NLQC protocol for generic $(1+1)$-dimensional local quantum systems.
In \cref{sec: simulating-holographic-cfts} we introduce a notion of a simulation of the CFT that replicates a minimal set of features to guarantee that the simulation ``captures'' the NLQC.
In \cref{sec: NLQC-via-holographic-states} we apply the many-body NLQC protocol to the CFT simulation and show that the computation it implements is equal to the bulk dynamics of interest.
In \cref{sec:computations}, we discuss which computations can be realized as bulk dynamics, and under an assumption about feasability of quantum computation in AdS$_3$, we argue that any polynomially complex unitary can be implemented with polynomial entanglement.
In \cref{sec: toy-model} we give the toy model, and in \cref{sec: relation-to-previous-work} we discuss the connections to previous work.
Finally in \cref{sec: discussion} we discuss implications for holography, NLQC, and future directions.

\section{Background } \label{sec: background}
\subsection{Non-local quantum computation }\label{subsec: NLQC}

A non-local quantum computation (NLQC) task involves two parties, Alice and Bob.
At the start, Alice holds quantum systems $A$ and $R_A$ while Bob holds $B$ and $R_B$.
These start out in a state $\ket{\psi}_{ABE}\otimes\ket{\phi}_{R_AR_B}$.
The state $\ket{\psi}_{ABE}$ is unknown to Alice and Bob, but they are free to choose the ``resource'' state $\ket{\phi}_{R_AR_B}$.
$E$ is a reference system which is unavailable to either party.
Alice and Bob are given an isometry $V_{AB\rightarrow\tilde{A}\tilde{B}}$\footnote{This can be generalized to any CPTP map, but we work with isometries for notational convenience.} acting on the system $AB$, which they are to perform in the following restricted matter. First, Alice performs a quantum channel of her choice $\mathcal{N}_{AR_A\rightarrow K_AM_A}$ on the systems she holds, producing a message system $M_A$ to be sent to Bob and a system $K_A$ that she keeps.
Similarly Bob performs $\mathcal{N}_{BR_B\rightarrow K_BM_B}$. Alice and Bob then exchange $M_A$ and $M_B$, but critically, these two systems do not interact during the exchange, which is why the computation is ``non-local''. Alice then performs a channel $\mathcal{N}_{K_AM_B\rightarrow \tilde{A}}$ and Bob the channel $\mathcal{N}_{K_BM_A\rightarrow \tilde{B}}$. They succeed in the task if the final state of $\tilde{A}\tilde{B}E$ is $V_{AB\rightarrow\tilde{A}\tilde{B}}\ket{\psi}_{ABE}$. This procedure is illustrated in \cref{fig:NLQC}.

\begin{figure}[t]
\centering
\includegraphics[width=0.6\textwidth]{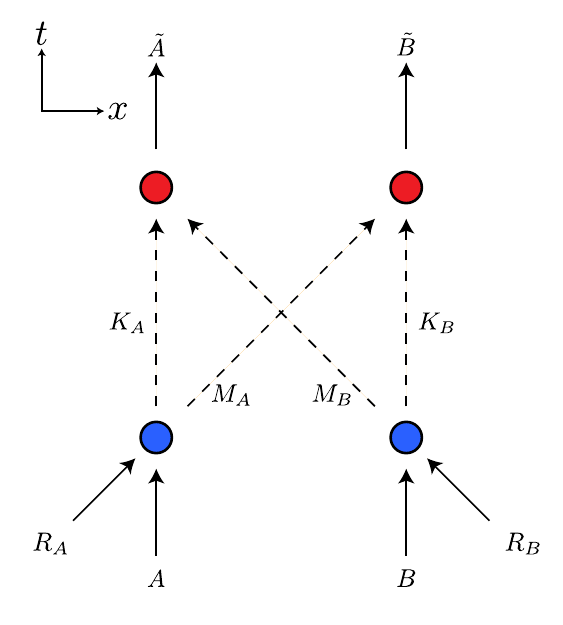}
\caption{A non-local quantum computation (NLQC) task. The circles on the left and right denote points at which Alice and Bob act respectively.}
\label{fig:NLQC}
\end{figure}

\subsection{AdS/CFT}
\label{subsec: AdS-CFT}

AdS/CFT is a general family of holographic dualities which relate a bulk asymptotically AdS theory of quantum gravity in $D+1$ dimensions to a CFT in $D$ dimensions living on its boundary. The gravity theory is referred to as the ``bulk'' while the CFT is referred to as the ``boundary''. A pedagogical introduction to the subject can be found in Ref.~\cite{Harlow-TASI}, and previous work on its connection with non-local computation is also generally accessible \cite{Alex-tasks,Alex-connect-wedge-theorem,Alex-complexity}. The aspects of AdS/CFT we need in order to make the connection with NLQC are 1) entanglement wedge reconstruction, which says how bulk subregions are encoded into boundary subregions, and 2) the causal structure discrepancy between bulk and boundary mentioned above, which motivates the connection. 

\subsubsection{Entanglement wedge reconstruction} \label{sss:EWR}

The AdS/CFT correspondence states that expectation values of bulk observables can be mapped to those of boundary observables.
For any boundary state with a well-defined dual bulk geometry, subregions of that bulk geometry are dual to subregions of the fixed boundary geometry.
Specifically, any spatial boundary region $R$ can reproduce
any correlation functions of at most an $O(1)$ number of field excitation operators supported in a corresponding bulk region known as its \emph{entanglement wedge}, $E(R)$. 

The fact that the bulk geometry is state-dependent means that this boundary operator is only guaranteed to have a local bulk interpretation for states that are similar to the initial one, i.e.\ if they have the same or similar bulk geometry.
The region $E(R)$ is best understood using the maximin prescription, see Ref.~\cite{Wall_2014,geoff} for more details.
Technically, it is a spacetime region that can be defined as the domain of dependence of a particular spatial bulk region -- that is, a $(2+1)$-dimensional spacetime region; however, we just refer to the spatial region as the entanglement wedge for simplicity. 

For our purposes we just note the following features (see also \cref{fig:entanglement-wedge-reconstruction}):
\begin{itemize}
\item The component of the boundary of the entanglement wedge $E(R)$ that lives on the boundary of the theory is precisely $R$.
\item For states near vacuum AdS, the entanglement wedge of any contiguous half of the boundary is the half of the bulk given by its convex hull, i.e.\ the half disk determined by the half circle.
\end{itemize}

Using entanglement wedge reconstruction, Ref.~\cite{Almheiri-Harlow} showed that the CFT acts as an error-correcting code from bulk to boundary with information living in $E(R)$ robust against erasures of $R^c$.
Rather than one fixed code space, however, a code space is determined by which bulk operators are probed.
The code space is created by acting on a semiclassical CFT state $\ket{\Omega}_{CFT}$, e.g.\ the ground state, with the boundary operators that reconstruct bulk operators of interest. 

As an example, consider a scenario in which a low-energy particle capable of storing a qubit follows a localized trajectory in the bulk as in \cref{fig:entanglement-wedge-example}.
Suppose we would like to give a boundary description of the particle along a particular point in its trajectory for any state of the qubit.
We may pick a Cauchy slice that intercepts this point, and consider the state of the boundary $\ket{\psi}$ at that time, when the qubit is set to some state.
Suppose the point is in the entanglement wedge of a boundary region $R$.
For each operator $X^i\in L(\mathbb{C}^{2})$ acting on the qubit, there is a bulk operator $X^i_{blk}$ localized at this point that implements it.
This in turn may be realized as a boundary operator $X^i_{CFT}$ supported on $R$.
The codespace is then given by $\text{span}\{X^i_{CFT}\ket{\psi}\}$, and there exists an encoding isometry $V:\mathbb{C}^2\rightarrow\mathcal{H}_{CFT}$ whose image is this code space and such that $VX^i=X^i_{CFT} V$.

\begin{figure}
\centering
\includegraphics[width=0.5\linewidth]{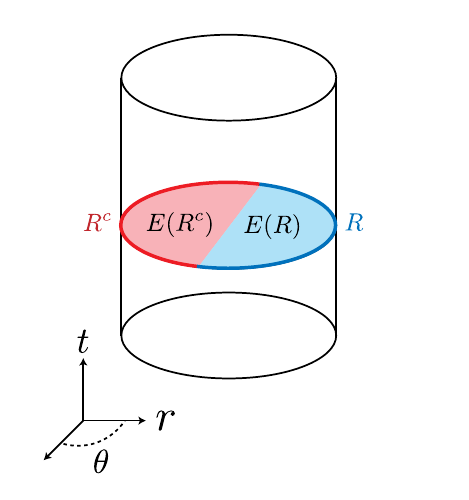}
\caption{Entanglement wedge reconstruction for vacuum AdS. A boundary spatial subregion $R$ (blue line) is dual to its entanglement wedge, a bulk spatial subregion (blue region) $E(R)$. Since $R$ is half of the boundary, $E(R)$ is half of the disk. Because we have specified a bulk foliation (constant $t$ in cylindrical coordinates), we can sensibly speak of a single time coordinate whose evolution corresponds simultaneously to both the bulk and boundary, which we henceforth do.}
\label{fig:entanglement-wedge-reconstruction}
\end{figure}

\begin{figure}
\centering
\includegraphics[width=0.3\linewidth]{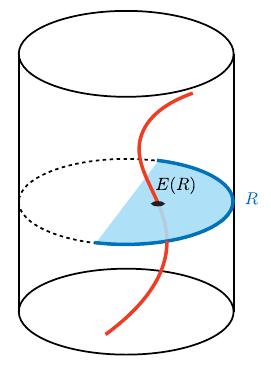}
\caption{Bulk trajectory of a system capable of storing quantum information. We may consider the depicted Cauchy slice as an error-correcting code protecting against the erasure of $R^c$ which encodes the state of the particle at the spacetime point depicted by the black dot.}
\label{fig:entanglement-wedge-example}
\end{figure}

\subsubsection{Causal structure discrepancy}\label{subsubsec:causal-structure-discrepancy}

One striking feature of this duality is a discrepancy between the causal structures of the bulk and boundary spacetime. Let $p,q$ be points in a spacetime. Define $p\prec q$ to mean that there exists a causal curve from $p$ to $q$ within that spacetime. 
Let $J^+(p)\equiv\{q| p\prec q\}$ and $J^-(p)\equiv\{q| q\prec p\}$ be the future and past of point $p$ respectively.
These are the set of all bulk points that can be affected by or that affect the point $p$, respectively.
Let $J^\pm$ refer to these sets in the bulk spacetime and $\hat{J}^\pm$ in the boundary spacetime.
The simplest example of this discrepancy occurs for the case with a $(1+1)$-dimensional boundary.
One may choose four spacetime points on the boundary, $c_0,c_1,r_0,r_1$, such that \cite{Heemskerk_2009,Gary_2009,Maldacena-Duffin}
\begin{align}
P&\equiv J^+(c_0)\cap J^+(c_1)\cap J^-(r_0)\cap J^-(r_1)\neq \varnothing \label{eq: bulk-interaction-region}\\
\hat{P}&\equiv \hat{J}^+(c_0)\cap \hat{J}^+(c_1)\cap \hat{J}^-(r_0)\cap \hat{J}^-(r_1)= \varnothing.\label{eq: boundary-interaction-region}
\end{align}
This is illustrated in \cref{fig:causal-structure-discrepancy}. 

\begin{figure}
\centering
\begin{subfigure}[t]{0.30\textwidth}
\centering
\includegraphics[width=\textwidth]{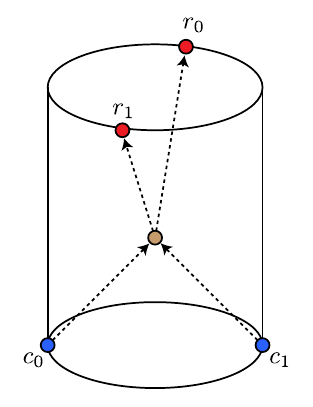}
\caption{The set of possible local interaction points $P \neq \varnothing$, i.e.\ two observers starting at $c_0$ and $c_1$ can meet up in the bulk, then return to $r_0$ and $r_1$.
}
\label{fig: bulk-task}
\end{subfigure}
\hfill
\begin{subfigure}[t]{0.57\textwidth}
\centering
\includegraphics[width=\textwidth]{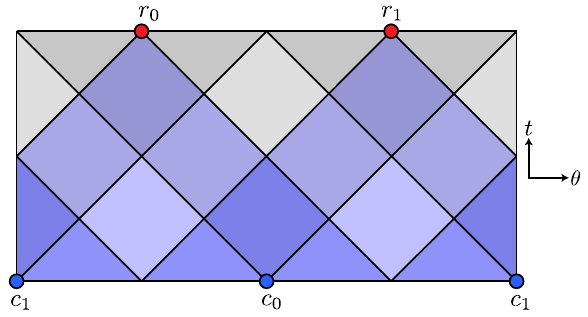}
\caption{The region of possible local interaction points $\hat{P}=\varnothing$, i.e.\ when constrained to the boundary these same observers cannot meet up before reaching $r_0$ and $r_1$.}
\label{fig: boundary-task}
\end{subfigure}
\hfill
\caption{Difference between bulk and boundary causal structure. Information starting at $c_0$ and $c_1$ and ending at $r_0$ and $r_1$ can interact locally in the bulk, but the dual boundary description must be non-local in some sense.}

\label{fig:causal-structure-discrepancy}
\end{figure}

\section{Non-local computation in generic local $(1+1)$-dimensional systems }\label{sec: many-body-nlqc}

In this section we introduce a class of NLQC protocols inspired by AdS/CFT, but which applies more broadly to general quantum systems that live on circular lattices.
Given such a system, the protocol performs some computation non-locally using the initial state as a resource. 
For any given computation there exists a system such that this protocol implements it\footnote{This is done by embedding the general purpose NLQC protocol of Ref.~\cite{beigi-koenig} directly into the Hamiltonian.}. 
Another interesting question is for a family of systems, such as those with a given Hilbert space dimension, what corresponding family of computations can be non-locally performed in this way.
Later we consider the family of systems that are \emph{holographic}, where we have more to say about this class.

\subsection{Many-body NLQC}

We present the precise claim here before elaborating on the relevant objects in the remainder of the section.
Essentially, we extract an NLQC protocol for any combination of initial state and dynamics for a many-body system on a circle, and any choice of how to encode/decode the input/output NLQC registers into/out of that many-body system.

\begin{figure}
\centering
\includegraphics[width=0.7\textwidth]{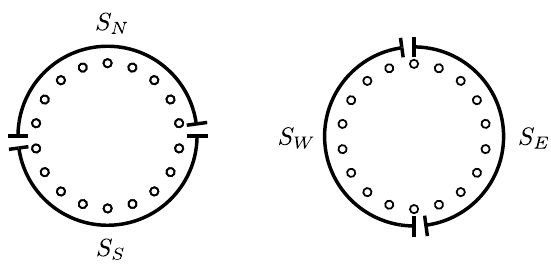}
\caption{A quantum many body system living on a circle. Each site (small circle) is associated with a register $S_\phi$ with $\phi$ an angular coordinate, and these are grouped into two different partitions. The labels $N,S,W,E$ are named after the directions on a compass.}
\label{fig:circle-regularization}
\end{figure}

\begin{definition}[Pseudo-bulk dynamics] \label{def:pbd}
Given the following objects,
\begin{itemize}
\item A quantum many-body system $S$ on a circle with $d$-dimensional subsystems labelled by angular coordinates as $S_\angcutoff,S_{2\angcutoff},...,S_{2\pi}$,
\item Two partitions of $S$, $\west \east$ and $\north \south$, where (see also \cref{fig:circle-regularization}) 
\begin{align*}
\west &= S_{[-\pi,0)}\\
\east &= S_{[0,\pi)}\\
\north &= S_{[-\pi/2,\pi/2)}\\
\south &= S_{[\pi/2,-\pi/2)}
\end{align*}
and by $S_{[\phi_1,\phi_2)}$ we mean the collection of $S_\phi$ systems such that $\phi\in [\phi_1,\phi_2)$ and all angles are defined modulo $2\pi$,
\item A unitary $U$ that acts on the full system $S$,
\item An initial state $\ket\psi$ of the system $S$,
\item A pair of $n/2$-qubit systems $A$ and $B$,
\item Encoding maps $\mathcal{N}_{A\west \to \west}$ and $\mathcal{N}_{B\east \to \east}$ (both CPTP),
\item Decoding maps $\mathcal{D}_{\north \to \tilde{A}}$ and $\mathcal{D}_{\south \to \tilde{B}}$ (both CPTP),
\end{itemize}
we define the \emph{pseudo-bulk dynamics} map as
\begin{equation}
\mathcal{B}(\rho_{AB}) = \mathcal{D}_{\north \to \tilde{A}}\otimes \mathcal{D}_{\south \to \tilde{B}} \left[ U\bigg(\mathcal{N}_{A\west\rightarrow\west}\otimes\mathcal{N}_{B\east\rightarrow\east}(\rho_{AB} \otimes \ketbra{\psi}_S)\bigg)U^\dagger\right].
\end{equation}
\end{definition}
This map is clearly CPTP. We return shortly to the discussion of appropriate choices for encoding and decoding maps.
First, one additional constraint on $U$ is needed before we can extract an NLQC protocol that implements these pseudo-bulk dynamics -- a notion of locality we refer to as the \emph{approximate spread} of a unitary, which measures the spatial extent of its causal influence.
Let the \emph{spread} of a unitary $U$ be the smallest distance $s$ such that for any local operator $A$ at site $S_\phi$, the support of $UAU^\dagger$ is contained in the interval $(\phi-s,\phi+s)$.
Typically we think of the dynamics $U$ as arising from some one-parameter family of unitaries $U(t)$, specifically generated by a time-independent Hamiltonian.
Assuming a non-trivial Hamiltonian, the spread immediately saturates to its maximum value of $\pi$ for any $t>0$; however the results of Lieb and Robinson \cite{Lieb-Robinson} provide an approximate notion of spread that grows much more steadily.
An equivalent notion to an operator $O$ being supported within an interval is that it commutes with any operator $O'$ supported \emph{outside} this interval, i.e.\ $[U(t) O U^\dagger(t),O']=0$.
\emph{Approximate spread} just amounts to relaxing the support requirement to only demand that all such commutators are small in the operator norm\footnote{Specifically, it would be parameterized by some $\epsilon$ that upper bounds these norms.}. 
The approximate spread increases linearly with $t$, and for the remainder of this section, we normalize the time co-ordinate such that this linear coefficient (the Lieb-Robinson velocity) is one.

\begin{claim}
Consider a set of objects for which the pseudo-bulk dynamics is defined.
Then as long as the approximate spread of $U$ is at most $2\pi/8$, there is an NLQC protocol implementing the pseudo-bulk dynamics using the initial state $\ket\psi_S$ as the resource state, with Alice holding $\west$ and Bob holding $\east$.
\end{claim}
The bound on approximate spread is crucial because, as we soon show, it allows us to decompose the unitary as 
\begin{align}\label{eq: unitary-chopping}
U\approx U_\north U_\south U_\west U_\east
\end{align}
with the subscripts denoting the support of the unitaries.
With this decomposition, the protocol demonstrating the claim is rather straightforward.
Alice and Bob each use their respective encoding maps, $\mathcal{N}_{A\west\to \west}$ and $\mathcal{N}_{B\east\to \east}$.
Alice then acts on $\west$ with $U_\west$ while Bob acts on $\east$ with $U_\east$.
They then use their one round of communication to exchange systems so that Alice holds $\north$ and Bob holds $\south$.
Alice acts on $\north$ with $U_\north$ while Bob acts on $\south$ with $U_\south$.
Thus the unitary $U$ has approximately been applied to $S$.
They then apply their respective decoding maps $\mathcal{D}_{\north\to \tilde{A}}$ and $\mathcal{D}_{\south\to \tilde{B}}$.

Let us now return to the subject of choosing encoding and decoding maps. When we apply this construction to AdS/CFT, the constraint that the computation equal the bulk dynamics determines the maps. However, even for non-holographic systems a similar choice of maps has a natural interpretation. Suppose that that each $S_\phi$ has a subsystem $I_\phi$ whose dimension matches that of $A$ and $B$, and a (possibly overlapping) subsystem $F_\phi$ whose dimension matches that of $\tilde{A}$ and $\tilde{B}$. As an encoding map Alice can apply $\text{SWAP}_{AI_{-\pi/2}}$ and then trace out the $A$ system.
Bob can apply a similar operation at $\pi/2$.
To decode, they swap the $F$ systems at $0$ and $\pi$ with $\tilde{A}$ and $\tilde{B}$ respectively.
The pseudo-bulk dynamics then tell us how probes at the points $\pi/2$ and $-\pi/2$ (analogous to $c_0$ and $c_1$ in the previous section) affect the points $0$ and $\pi$ (analogous to $r_0$ and $r_1$), much like a four point function.
Alternatively, one could choose the decoding map to be whichever one yields the pseudo-bulk dynamics with the highest capacity, i.e.\ define the systems $\tilde{A}$ and $\tilde{B}$ such that they maximize information from $A$ and $B$.

Alice and Bob in fact have a bit more flexibility. If $U(t)$ has an approximate light-cone, they can effectively increase the time for which the system can evolve by applying some of the dynamics locally. For example, Alice can apply $U(t') \text{SWAP}_{AI_{-\pi/2}} U^\dagger(t')$ with $t'<\pi/2$ which by the definition of the light-cone property of $U(t')$, is still supported in Alice's region $\west$\footnote{If the light-cone is approximate Alice can truncate appropriately as we discuss more in \cref{subsec: chopping-boundary}.}.
This encoding map is the one we focus on in \cref{sec: NLQC-via-holographic-states}, as it also has a nice holographic interpretation.
Consider an initial state with a system near the centre of the bulk.
The above encoding channel amounts to rewinding time so that the system is at the boundary, swapping it with the input system $A$, and then letting it fall back towards the centre again. Alice can similarly apply $U^\dagger(t') \text{SWAP}_{AF_{0}} U(t')$ to gain additional time at the end.

\subsection{Decomposing the boundary unitary }\label{subsec: chopping-boundary}

A key step in showing that the CFT performs NLQC is in showing that the boundary time evolution can be decomposed into four components as in \cref{eq: unitary-chopping} -- in fact, this was also implicitly required in the protocol extraction method described in Ref.~\cite{Alex-complexity}.
Shortly, we give an argument for such a decomposition directly in terms of the properties of the CFT (or in fact, any QFT).
Specifically, one can find local operations that deform Cauchy slices within their domain of dependence.
While some QFT simulations can directly capture this feature \cite{Osborne-Stottmeister}, it appears to be quite a stringent requirement due to the number of possible deformations; we believe that a simulation can accurately capture the basic physics relevant for NLQC without preserving this more intricate feature. 
We then introduce an alternative decomposition using a property that should be easier for a simulation to capture -- namely that global time evolution does not allow superluminal signalling. 
Whether the former property implies the latter is an interesting question that we revisit in the discussion.

\subsubsection{Decomposition via Cauchy slices}
\begin{figure}[h]
\centering
\includegraphics[width=1.0\textwidth]{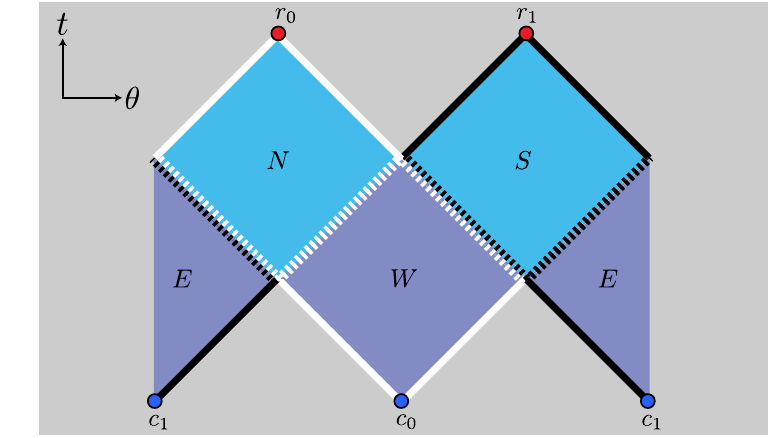}
\caption{Sketch of a many-body NLQC protocol using Cauchy slice deformation. White lines denote systems that Alice holds, while black lines denote systems that Bob holes. Dotted lines denote message systems, and change color to indicated they have been exchanged.}
\label{fig:our-setup}
\end{figure}
If we are dealing with a QFT, we can obtain a decomposition of the form \cref{eq: unitary-chopping} directly by deforming Cauchy slices. 
Consider the Cauchy slice depicted in \cref{fig:our-setup} collectively by the lower solid zigzag line.
As we make more precise in \cref{sec: simulating-holographic-cfts}, the white and black segments of this line be though of as a finite-dimensional quantum systems, which we call $W$ and $E$ respectively.
As the resources, give Alice $R_A=W$ and Bob $R_B=E$. Alice and Bob apply some encoding, such as swapping $A$ and $B$ into a subsystem on $c_0$ and $c_1$ respectively. The unitary $U$ decomposes as $U_NU_SU_WU_E$, where $U_W$ and $U_E$ evolve the lower white and black line to the dotted white and black lines respectively. Similarly, $U_N$ and $U_S$ evolve the dotted white and black lines into the upper white and black lines respectively. Alice and Bob then apply some encoding, such as swapping $\tilde{A}$ and $\tilde{B}$ out of subsystems on $r_0$ and $r_1$ respectively.

\subsubsection{Decomposition via limited spread}\label{subsubsec: decomp-via-limited-spread}
We now prove \cref{eq: unitary-chopping} without using local Cauchy deformations, and instead by just using the fact that the spread of the boundary dynamics $U$ is limited.
We find that the maximum spread such that $U$ can be decomposed in this way is $s=2\pi/8$. 
 
We also explain how to extend the result to approximate spread in \cref{appendix: err-from-lightcone}.

\begin{figure}[h]
    \centering
    \includegraphics[width=1\linewidth]{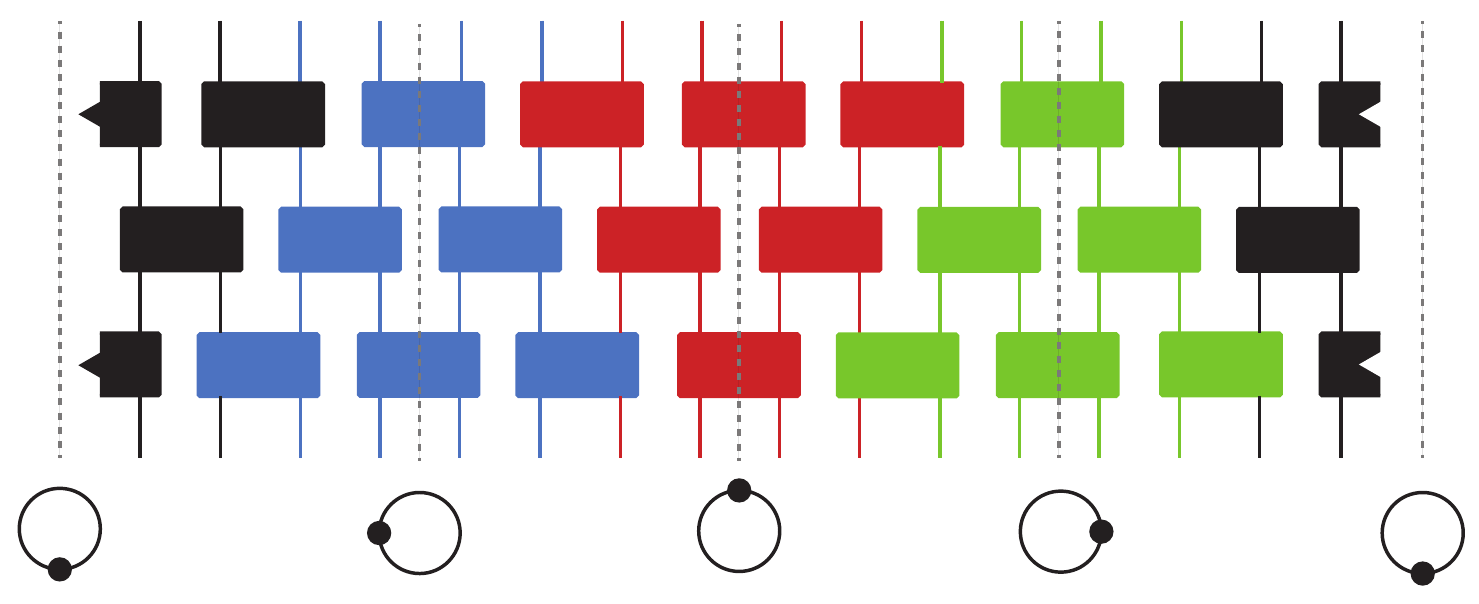}
    \caption{Decomposition of any finite depth circuit whose structure imposes that information spreads moves by less than an eight of the circle into the form of \cref{eq: unitary-chopping}. The blue and green sub-circuits make up $U_W$ and $U_E$ respectively, while the red and black sub-circuits make $U_N$ and $U_S$ respectively. However not all unitaries with this limited information spread can be decomposed in such a manner.}
    \label{fig:circuit-method}
\end{figure}

Consider a unitary that comes from a family parametrized by time, such as a Hamiltonian evolution or a discrete quantum circuit.
Ideally $U(t)$ could be decomposed as a local circuit whose structure makes its spread inherently clear.
In that case we could easily decompose it as desired, as shown in \cref{fig:circuit-method}, so long as $s\leq 2\pi/4$.

However, it is known that not all unitaries admit circuits that make their causal properties clear in this manner \cite{gross}.
Using just the spread is enough, provided $s\leq 2\pi/8$, since we can use the technique of Ref.~\cite{unitarity-plus-causality} to decompose the unitary in the following manner.
Introduce an identical system $S'$ in some fixed state $\ket{0}_{S'}$, which serves as an auxiliary system to help with implementing the decomposition of the time evolution $U$.
As long as a decomposition of the form in \cref{eq: unitary-chopping} implements $U_S \otimes U'_{S'}$ for some unitary $U'$, then the additional system $S'$ does not affect the validity of the NLQC protocol.
The primed system has corresponding subsystems such as  $S'_W$ or $S'_{\phi}$ for any $\phi \in [0,2\pi)$, and the components of our decomposition now have support on e.g.\ $\west S'_W$.

We can implement such a unitary -- specifically, $U_S \otimes U^\dagger_{S'}$ -- as follows.
Let $\Sigma _{\phi}$ be the operator that swaps $S_\phi$ and $S'_\phi$.
Define the operators
$
    K_\phi=U_{S'}\Sigma_\phi U^\dagger_{S'}
$, which mutually commute because the $\Sigma_\phi$ do.
The product of all of these simply evolves $S'$ backwards with  $U^\dagger$, swaps the systems $S$ and $S'$, and then evolves $S'$ forwards with $U$.
Thus, we find
\begin{align*}
    \prod_{\phi\in [0,2\pi)}\Sigma_\phi\prod_{\phi\in [0,2\pi)}K_\phi =U_{S} \otimes U^\dagger_{S'},
\end{align*}
and so it is sufficient to implement the operator on the left hand side.
Note that the $K_\phi$ have support of about $2s$ radians. If $s\leq 2\pi/8$, then this is at most a quarter of the circle.
But this means it is always contained in at least one of $\north,\south,\east,$ or $\west$ and its primed counterpart. Thus we can rewrite 
\begin{align*}
     \prod_{\phi\in [0,2\pi)}K_\phi=K_NK_SK_WK_E.
\end{align*}

We can then absorb the part of $\prod_{\phi\in [0,2\pi)}\Sigma_\phi$ that acts on $\north S'_N$ into $K_N$, and similarly for $K_S$, thus giving the desired decomposition.

When replacing spread with approximate spread, these $K$ operators may have some small components outside of their regions of support (e.g.\ exponential tails in the case of Hamiltonian evolution).
In that case, we simply truncate the tails as an approximation, such that the above product only approximately implements\footnote{Even if the truncated operator is not unitary, it is possible to construct a CPTP map that approximately implements it.} $U_S \otimes U'_{S'}$
Thus, evolution of any local Hamiltonian for time $t \leq \pi/4$ -- or any unitary with approximate spread $s \leq \pi/4$ -- can implement a unitary of the form in \cref{eq: unitary-chopping}.
This is done in \cref{app:spread}.


\section{Simulating holographic CFTs } \label{sec: simulating-holographic-cfts}
We would like to analyze the NLQC protocol introduced in \cref{sec: many-body-nlqc} for the particular case of a holographic CFT.
However there are two issues.
First, the techniques we used to derive the protocol are applicable only to finite-dimensional systems.
Second, the entanglement present between two halves of the CFT is infinite, which would naively indicate that such a protocol would be prohibitively costly regardless. 
We can solve both of these problems by considering finite-dimensional approximations of the CFT which preserve its relevant features.
This clearly solves the first issue, and solves the second by giving a finite upper bound on the entanglement used. 

To be precise about what we mean by an approximation of AdS/CFT we phrase the approximation in terms of a simulation running on a quantum computer.
There are three core features the simulation must capture in order to facilitate a meaningful NLQC protocol: geometric structure, dynamics, and no superluminal signalling.
By geometric structure we mean that operators representing observables of interest have counterparts in the simulation with similar geometric support.
By dynamics we mean that correlation functions of said operators are accurately captured.
Finally, by no superluminal signalling we mean that the support of a time-evolved operator grows no faster than the speed of light, even for simulation operators with no counterpart in the CFT\footnote{%
A simple simulation would have every operator correspond to some counterpart in the CFT.
However, to be more general, we allow for the possibility that only some subspace of the simulation Hilbert space represent the relevant CFT subspace.
Otherwise, the lack of superluminal signalling would be implied by the accuracy of correlation functions.
}.

In this section we give physical intuition for these features, and why we need them.
We then give a formal definition, and finally we argue why such a simulation should exist. 

\subsection{Intuition \& necessity of simulation features }

A quantum field theory is an idealized quantum system that assigns an infinite-dimensional degree of freedom to each point in space.
In reality, however, quantum field theories only model phenomena that occur above a given length scale, and when they are sufficiently well behaved, this makes it possible to obtain the same results by replacing the QFT with a more tractable \emph{regularized} version with a minimum spatial resolution or ``cutoff''.
A common method is lattice regularization, in which the space (in our case a circle) is modeled by a lattice with at most countably infinite-dimensional systems living on its cells and a local Hamiltonian acting on those degrees of freedom.
One strategy to simulate the QFT is thus to truncate the dimension of each of these degrees of freedom to a subset of states that is plausibly reached in a particular scenario of interest, e.g.\ Ref.~\cite{Preskill}.
Another is to look for a sequence of local finite memory systems of increasing resolution for which the correlation functions of course grained observables converge to those calculated directly from the QFT, e.g.\ Ref.~\cite{Osborne-Stottmeister}.

The most basic way to simulate a quantum system is to map its states of interest to states of the computer memory, and evolve by a Hamiltonian whose spectrum agrees for those states.
The simulation strategies above do this (with the states of interest being low-energy states), but they additionally preserve the geometric structure of the theory.
In particular, a field theory operator acting in a certain region of space is represented by a simulation operator acting on lattice sites in that region of space.

Why do we need this feature? In the many-body NLQC protocol described in \cref{sec: many-body-nlqc}, we needed a notion of locality (the various quantum systems were associated with points on a circle) and a time evolution with limited spread.
The time evolution of a quantum field theory already has an exact light-cone because it obeys special relativity.
The exact light-cone results in the spread of a local operator growing linearly with time.
Thus if we preserve the geometric structure of each operator's support, we can expect the simulation's time evolution to have limited spread for short enough times.

We must be careful, however, about the scope of this inherited light-cone.
To use the methods of \cref{sec: many-body-nlqc}, we needed either that Cauchy slices can be deformed with local operations, or that the spread restriction apply to a product of swap operators.
Given that some simulation operators may not have CFT counter parts, and of those which do  only a subset accurately capture correlation functions, the swap operators do not neccesarily inherit this restriction from the CFT.
Thus in order to use the limited-spread method, we require the additional condition that the light-cone applies to all operators.



Finally we require the simulation to capture the dynamics of the theory.
These are completely encoded into its correlation functions, so we insist that these are accurately captured.
However, once again it is too much to ask that any arbitrary correlation function be captured; any finite simulation has a limit to its resolution, so we require only that for any particular finite set of observables of interest, a simulation exists that can accurately capture reasonably simple correlation functions involving those observables.



\subsection{Formal definition }\label{subsec: formal-definition-simulation}

We now give a more formal definition of what we require from a simulation.
For the following specifications,
\begin{itemize}
    \item $\mathcal{H}$ and $H$, the CFT Hilbert space and Hamiltonian respectively,
    \item a state $\ket{\psi}\in \mathcal{H}$,
    \item a finite-dimensional vector space\footnote{That this is a vector space and not an algebra is due to to the interpretation of these operators as course grained observables, as we will soon discuss in more detail in the next subsection. In particular the adjoint of a coarse graining map preserves a vector space structure but not an algebraic one.} of operators $\vecspace$ acting on $\mathcal{H}$ e.g.\ the ones corresponding to observables at a given scale,
    \item an error tolerance $\delta$,
    \item a time $T$ representing the duration of time for which the simulation is accurate,
    \item a positive integer $m$ specifying how many operators may appear in a correlation function, and
    \item a finite set of times $\{t_k\}_{k=1}^{m}$ with $0\leq t_k<T$ for correlation functions to be evaluated at.
\end{itemize}
we define a simulation as follows.
A simulation consists of
\begin{itemize}
    \item $\mathcal{H}'$ and $H'$, the simulation Hilbert space and Hamiltonian respectively, with the Hilbert space constructed on a lattice regularization of the space on which the QFT is defined (in our case -- a circle -- there will be some small cutoff angle $\angcutoff$, with the simulation memory consisting of $\frac{2\pi}{\theta}$ qudit registers $S_\angcutoff,S_{2\angcutoff},...,S_{2\pi}$),
    \item a simulation state $\ket{\psi'}$,
    \item a finite-dimensional vector space of operators $\vecspace'$ acting on $\mathcal{H}'$, and
    \item a CPTP map $\mathcal{E}$ from density matrices on $\mathcal{H}$ to ones on $\mathcal{H}'$ whose adjoint $\alpha \defi \mathcal{E}^\dagger$ satisfies $\alpha(\vecspace')=\vecspace$%
    \footnote{Note that this is equivalent to requiring that for each $A\in\mathcal{M}$ there is a corresponding element $B\in \mathcal{M}'$ such that for any state $\rho$ of the CFT, $\ev{A}_\rho = \ev{B}_{\Lambda(\rho)}$; in other words the simulation accurately captures expectation values of this limited set of observables.}%
    ,
\end{itemize}
such that the following hold:
\begin{itemize}
    
    \item for any $m$ operators $\{M'_i\} \subset \mathcal{M}'$, and $m$ times $\{s_i\}$ with $s_i\in\{t_i\}$, we have that
    \begin{align}
       \label{eq: CFT-sim-correlation-functions-relationship} |\ev{M_{1}(s_1) \cdots M_{m}(s_m) }-\ev{M'_{1}(s_1) \cdots  M'_{m}(s_m) }|\leq \delta ||M'_{1}(s_1) \cdots  M'_{m}(s_m) ||
    \end{align}
    where $M_i:=\alpha(M'_i)$, and $M_i$ operators are time evolved using $H$ while $M'_i$ are time evolved with $H'$, and the correlation functions are taken with respect to $\ket{\psi}$ and $\ket{\psi'}$,
    
    
    \item  for any operator $\vecspaceop'\in \vecspace'$, the lattice points on which $\vecspaceop'$ is supported are contained within the region of space in which $\alpha(M')$ is supported, and
    \item for any local simulation operators $O$ and $O'$ and any $t<T$, 
    the approximate light-cone condition holds, i.e.\
    \begin{align}
        \| [O(t),O'] \| \leq  \|O(t)\| \|O'\| a \exp(-b ( d(\supp(O),\supp(O')) - c t) ),
    \end{align}
    where $a$ and $b$ are constants, $d(R,S) = \min_{x\in R, y\in S} d(x,y)$ is the distance between regions, and $c$ is the speed of light in the CFT (which we now set to one).
\end{itemize}
In general, smaller error parameters $\delta$ and $a$ (as well as larger decay constant $b$) will require an increasingly large simulation.
We define the resolution scale as $N_S = \log(\text{\# of sites})$, i.e.\ in our case $N_S = \log(\frac{2\pi}{\theta})$.
For some of the results in this paper to hold, we will additionally require that the simulation be \emph{efficient}, in that this resolution scale does not need to increase too quickly to suppress the errors:
\begin{itemize}
    \item the error scales as $\delta\sim \text{poly}(N_S^{-1})$, where  is the resolution scale of the simulation, and
    \item $a$ increases at most exponentially in $N_S$,
    \item $b$ increases superlinearly with $N_S$.
\end{itemize}

\subsection{Existence of the simulation}

A variety of quantum simulation techniques for QFTs have been studied (see for example \cite{Preskill,Buser:2020cvn,Osborne-Stottmeister,Byrnes:2005qx} among others).
None of these have been shown to work for strongly-coupled CFTs such as those from holography. 
Nevertheless, they display many or all of our required features for a simulation.
If a suitable simulation technique is devised for holographic CFTs, we thus find it reasonable that it might satisfy our requirements.

The feature of preserving geometric locality is satisfied by many simulation techniques, usually via real space regularization.
The accuracy of estimating correlation functions via these techniques has been rigorously demonstrated for a number of QFTs, though not for strongly-coupled ones such as holographic CFTs.
However, under the assumption of the legitimacy of lattice regularization such as the widely used formulation for QCD introduced by Wilson \cite{Wilson,kogut-susskind}, there exist simulation techniques for strongly-coupled systems that preserve these features as well \cite{Byrnes:2005qx}.
This assumption is supported by simulations on classical computers, which have used the lattice formulation successfully to predict experimentally verifiable results such as the mass of the proton \cite{Durr:2008zz}.

The light-cone requirement enforces the causal structure of the CFT in the simulation.
Such a light-cone behaviour emerges in lattice systems via the Lieb-Robinson velocity \cite{Lieb-Robinson,hastingsLR,Hamma_2009}.
If this velocity approaches the speed of light, then the condition will be satisfied.
For example, Ref.~\cite{Hamma_2009} shows a lattice system that approaches an electromagnetic theory in the continuum limit, and finds that the Lieb-Robinson velocity is upper bounded by $O(c)$, slightly greater than the speed of light.
The upper bound is loose and can likely be improved with better techniques.
This example also meets our criteria for the scaling of error parameters associated with the approximate light-cone.
The coefficient $a$ from the Lieb-Robinson bound scales exponentially in $N_S$, as it is proportional to the number of sites in a fixed geometry.
However, the decay parameter is proportional to the inverse lattice scale, meaning that it increases exponentially in $N_S$.
So long as it increases superlinearly, this offsets the increase in $a$ as one improves the simulation resolution.
We expect such behaviour to be typical for simulations of relativistic theories.

The error scaling of $\delta$ just means that we only need a number of simulation qubits that is at most exponential in $\frac{1}{\delta}$ to attain error $\delta$.

Finally, we believe these requirements to be reasonable as we already know of simulations that satisfy them, albeit for non-holographic CFTs.
The quantum simulation technique of Ref.~\cite{Osborne-Stottmeister} simulates a particular non-holographic CFT, obeying the requirements of an efficient simulation described in \cref{subsec: formal-definition-simulation}\footnote{
This is with the exception of the light-cone condition, which remains unclear; however, it does satisfy the local Cauchy slice deformation condition.
Although we do not require the local Cauchy slice deformation condition in the formal definition above, we believe similar results to those presented here could be obtained using this condition in place of the no-superluminal-signalling condition.
}.

\section{Non-local computation via holographic states }\label{sec: NLQC-via-holographic-states}
We have now developed sufficient machinery to give a rigorous notion for how AdS/CFT performs non-local computation.
Given an initial state and encoding/decoding maps, the techniques of \cref{sec: many-body-nlqc} give an NLQC protocol for any $(1+1)$-dimensional many-body system, such as a quantum simulation of a $(1+1)$-dimensional QFT as described in \cref{sec: simulating-holographic-cfts}.
For the choices of initial state and encoding/decoding maps we describe in \cref{subsec:setup}, the NLQC protocol extracted in this way implements an isometry that is enacted holographically in the bulk as a local computation.
We prove this by first connecting the bulk computation to the boundary CFT time evolution in \cref{subsec:bulk-boundary}, then connecting to the CFT simulation in \cref{subsec:bulk-simulation}.
We more formally present the resultant many-body NLQC protocol and demonstrate that it implements the bulk computation in \cref{subsec:computation-validity}.
In \cref{subsec:error}, we demonstrate the robustness of this result to the presence of error from (i) the approximate nature of Alice's and Bob's local encodings, (ii) the approximate nature of their local recovery of information from the outputs, (iii) the imperfection of the dynamical duality, coming from both the approximate nature of holographic duality and from the imperfect accuracy of the simulation, and (iv) the approximate nature of the simulation light-cone.
Finally, in \cref{subsec:entanglement} we estimate the amount of entanglement consumed by the protocol, and find that it scales as $\frac{N_S}{G_N}$, with $N_S$ the resolution scale of the simulation (dependent only on the choices of error parameters, not on the number of input qubits $n$) and $G_N$ the Newton's constant of the holographic theory (which is decreased relative to a fixed reference length scale, associated with an increase in the number of qubits per site).

\subsection{Applying many-body NLQC to a holographic simulation} \label{subsec:setup}
Recall from \cref{sec: many-body-nlqc} that for a many-body system (such as the finite-dimensional CFT simulations described in \cref{sec: simulating-holographic-cfts}), we required a list of ingredients in order to extract an NLQC protocol.
This included a choice of initial state, time-evolution unitary, and encoding maps for specific systems $A$ and $B$ and decoding maps for systems $\tilde{A}$ and $\tilde{B}$.
As we show in the rest of the section, the pseudo-bulk dynamics arising from this procedure can be interpreted as the actual bulk computation when these parameters are chosen as follows.
The initial state at $t=0$ should be a simulation of a CFT state with a bulk dual in which two systems carrying quantum information, $A_L$ and $B_L$, are moving towards the center of the disk from opposite directions, say from the left and right (see \cref{fig:bulk-computation-and-cross-sections}).
The encoding map is one that encodes $A$ and $B$ into the simulation's analogue of these respective bulk systems.
After evolving with the unitary $U$ corresponding to simulation Hamiltonian time evolution with time $\tau$, 
we require that in the bulk, two other systems $\tilde{A}_L$ and $\tilde{B}_L$ are shot out from the center along the perpendicular direction, i.e.\ up and down respectively (see \cref{fig: bulk-dynamics}).
This could be, for example, due to a simple two particle scattering process, or as discussed in Ref.~\cite{Alex-connect-wedge-theorem}, a quantum computer living in the center that performs operations on the two systems before ejecting them.
The decoding maps are the simulation's entanglement wedge recovery maps (see \cref{sss:EWR}) for these registers at time $\tau$.

\begin{figure}
     \centering
     \begin{subfigure}[t]{0.46\textwidth}
         \centering
\includegraphics[width=\textwidth]{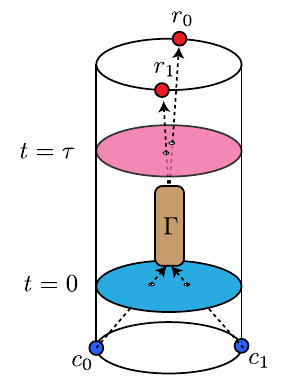}
         \caption{Quantum systems come in from points $c_0$ and $c_1$ and interact in the center via some computation (isometry) $\Gamma$, whose outputs are ejected to points $r_0$ and $r_1$. 
}
         \label{fig: bulk-dynamics}
     \end{subfigure}
     \hfill
     \begin{subfigure}[t]{0.42\textwidth}
         \centering
         \includegraphics[width=0.8\textwidth]{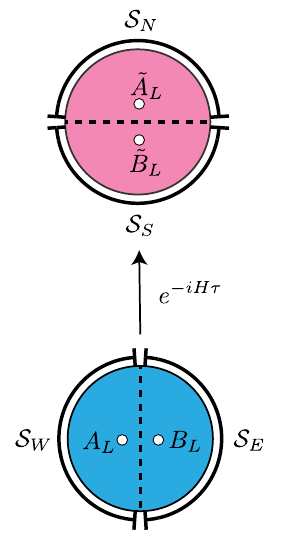}
         \caption{Bulk constant $t$ Cauchy slices right before and after the computation. $A_L$ and $B_L$ are reconstructable on $\mathcal{S}_W$ and $\mathcal{S}_E$ respectively. Similarly $\tilde{A}_L$ and $\tilde{B}_L$ are reconstructable on $\mathcal{S}_N$ and $\mathcal{S}_S$ respectively.}
         \label{fig: emergent-system-dynamics}
     \end{subfigure}
     \hfill
        \caption{A bulk process whose dynamics can be implemented non-locally using the boundary as a resource.}
        
        \label{fig:bulk-computation-and-cross-sections}
\end{figure}

The bulk information contained in $A_L$ and $B_L$ described above can be mapped isometrically to the boundary CFT system, which we refer to as $\Csys$.
Specifically, they can be mapped to regions that we denote $\Cwest$ and $\Ceast$, which are the subregions of the CFT corresponding to the simulation subregions $\west$ and $\east$ respectively.
This arises from the error-correction properties of holography, specifically sub-region duality (see \cref{subsec: AdS-CFT}), which provides a mechanism for recovering bulk local information on an appropriate boundary region.
In particular, we can treat the $t=0$ Cauchy slice as an error-correcting code with the following features.
The encoding isometry $V_0: \mathcal{H}_{A_LB_L}\to \mathcal{H}_{CFT}$ maps from the logical Hilbert space to states in the CFT, such that the system $A_L$ is recoverable on the left half of the CFT $\Cwest$ and the system $B_L$ is recoverable on the right half $\Ceast$\footnote{These error correction properties are only approximate but become exact in the limit $1/G_N\rightarrow\infty$.}.
This recoverability of the system $A_L$ is equivalent to the fact that for any operator acting on it, $O_{A_L}\in\mathcal{L}(\mathcal{H}_{A_L})$\footnote{Here we use $\mathcal{L}(\mathcal{H})$ to denote the set of linear operators acting on $\mathcal{H}$.}, there exists a CFT operator $O_{W}$ acting on the left half $\Cwest$ that commutes with the code space projector $V_0V_0^\dagger$ and also implements $O_{A_L}$ in the sense that $V_0O_{A_L}=O_{W}V_0$; similarly for $B_L$.
Similarly at time $t=\tau$, there will be an isometry $V_\tau: \mathcal{H}_{\tilde{A}_L\tilde{B}_L} \to \mathcal{H}_{CFT} $ such that $\tilde{A}_L$ is recoverable on the top half $\Cnorth$ and $\tilde{B}_L$ on the bottom half $\Csouth$ (see \cref{fig: emergent-system-dynamics}).

We introduce a bulk isometry\footnote{Again, we use an isometry here for notational convenience. Because there will always be some noise, the dynamics are better described by a channel, and our methods easily generalize to this case. } (or transition matrix) $\Gamma: \mathcal{H}_{A_LB_L}\rightarrow\mathcal{H}_{\tilde{A}_L\tilde{B}_L}$ to describe the bulk dynamics of these systems from time $0$ to time $\tau$ in the following sense.
Consider a basis $\ket{i}$ for $\mathcal{H}_{A_LB_L}$ and a basis $\ket{\tilde{j}}$ for $\mathcal{H}_{\tilde{A}_L\tilde{B}_L}$.
We assume these basis states are chosen such that they are product states with respect to the bipartite structure.
The vector $\ket{0_L}\in \mathcal{H}_{A_LB_L}$ represents the fixed initial state of the protocol, dual to $\ket{{\zc}}$ in the CFT, and we can define $\ket{\tilde{0}_L} = \Gamma \ket{0_L}$ as its time-evolved counterpart in the Schrodinger picture.
We can also introduce excitation operators $\Xl$ and $\tXl$ such that $\Xl^i\ket{0_L} = \ket{i}$ and $\tXl^j \ket{\tilde{0}_L} = \ket{\tilde{j}}$, and again note that these can be chosen to take the form of a tensor product of operators on the $A_L$ and $B_L$ systems (resp.\ $\tilde{A}_L$ and $\tilde{B}_L$).
Then the matrix elements of $\Gamma$ are given by
\begin{align}
    \bra{\tilde{j}} \Gamma \ket{i} &= 
    \bra{0_L} \Gamma^\dagger  \tXl^{\dagger j} \Gamma  \Xl^i \ket{0_L} .
\end{align}
Now, these finite-dimensional Hilbert spaces and their respective operators are really a simplified, emergent description of some more complicated bulk physics described by a QFT with semiclassical gravity -- much as any qubit in a lab emerges from the Standard Model.
In the full bulk theory, switching now to the Heisenberg picture, there are some operators $\tXb(\tau)$ and $\Xb$ such that:
\begin{align}
    \bra{0_L} \Gamma^\dagger \tXl^{\dagger j} \Gamma \Xl^i\ket{0_L} 
    &= \ev{ \tXb^{\dagger j}(\tau) \Xb^i } ,
\end{align}
where the correlation function is now to be evaluated according to the appropriate path integral. 
These operators $\Xb$ and $\tXb(\tau)$ are each tensor products of local operators in the bulk, acting in the corresponding regions of spacetime associated with emergent systems $\mathcal{H}_{A_LB_L}$ and $\mathcal{H}_{\tilde{A}_L\tilde{B}_L}$ respectively.
They retain a product structure across the regions associated with $A_L$ and $B_L$ (resp.\ $\tilde{A}_L$ and $\tilde{B}_L$). 

Such a correlation function is dual to one featuring some boundary operators $\Xc$ and $\tXc$, specifically
\begin{align}
    \ev{ \tXb^{\dagger j}(\tau) \Xb^i } 
    &= \bra{\zc} \tXc^{\dagger j}(\tau) \Xc^i \ket{\zc},
\end{align}
where these operators have the additional properties arising from the error-correction results of Ref.~\cite{Almheiri-Harlow} reviewed above%
, namely
\begin{align}
    \Xc V_0 &=  V_0 \Xl \label{eq:phi_i_reconstruct}\\
    \tXc V_\tau &=  V_\tau \tXl \label{eq:phi_j_reconstruct},
\end{align}
and $\tXc(\tau)$ is the Heisenberg picture version of $\tXc$, i.e.\ $\tXc(\tau) \defi U_\tau^\dagger \tXc U_\tau$, where $U_\tau$ is the boundary operator that implements local time-evolution by time $\tau$.
Furthermore, 
the $\Xc$ operator has a product structure between the boundary regions $S_W$ and $S_E$ \emph{at time zero},
according to the sub-region duality described above (see \cref{fig: emergent-system-dynamics}).
Similarly, the Schrodinger picture $\tXc$ operator has a product structure between $S_N$ and $S_S$, but this does not necessarily hold for its Heisenberg counterpart.

The above discussion justifies that the logical correlation functions match their CFT counterparts.
Of course, due to the imprecision of the holographic correspondence, this will only be up to some error.

For any set of $N$ logical operators $O_L^i$ acting on $\mathcal{H}_{A_LB_L}$ and any $t_1,\cdots,t_N \in \{0,\tau\}$, there exist CFT operators $O_{CFT}^i(t_i)$ such that the following holds,
\begin{align}\label{eq:logical-to-CFT-corr-err}
    | \bra{0_L} O_{L,t_1}^1 \cdots O_{L,t_N}^N \ket{0_L} -  \bra{0_{CFT}} O_{{CFT}}^1(t_1) \cdots O_{{CFT}}^N(t_N) \ket{0_{CFT}} | \leq O(G_N) \| O_L^1 \| \cdots \|O_L^N \|,
\end{align}
where $O_{L,0}^i \defi O_L^i$ and $O_{L,\tau}^i \defi \Gamma O_L \Gamma^\dagger$.

\subsection{Mapping the NLQC inputs/outputs to the boundary } \label{subsec:bulk-boundary}

Now, if we express the map $\Gamma$ as $\Gamma = \sum_{i,j} \gamma_{ij} \ketbra{\tilde{j}}{i}$, then the above implies that
\begin{align}
    \gamma_{ij} &= \bra{0_L} \Gamma^\dagger X^{\dagger j}_L  \Gamma X_L^i \ket{0_L} \\
    &=\bra{\zc} \tXc^{\dagger j}(\tau) \Xc^i \ket{\zc} \\
    &= \bra{\zc} U^\dagger _\tau \tXc^{\dagger j} U_\tau \Xc^i \ket{\zc} .
\end{align}
This tells us that the state $U_\tau \Xc^i \ket{\zc}$ has overlap $\gamma_{ij}$ with $ \tXc^j U_\tau \ket{\zc}$.
But since $\gamma_{ij}$ are the coefficients of an isometry $\Gamma$, and we can choose the operators $\Xc$ and $\tXc$ to be unitary (see \cref{lemma: unitary-reconstruction}), all other components of $U_\tau\Xc^i \ket{\zc}$ must vanish, i.e.
\begin{align}
    U_\tau \Xc^i \ket{\zc} = \sum_j \gamma_{ij}  \tXc^j U_\tau \ket{\zc}.
    \label{eq:decompose_tau}
\end{align}
Finally, we have
\begin{align}
    V_\tau \Gamma \ket{i} &=\sum_j \gamma_{ij} V_\tau  \tXl ^j\ket{\tilde{0}_L}\text{, by our decomposition of }\Gamma \\
    &=\sum_j \gamma_{ij} \tXc^j U_\tau\ket{{\zc}}\text{, by the reconstruction property of }\tXc,\text{ \cref{eq:phi_j_reconstruct}} \\
    &= U_\tau \Xc^i \ket{{\zc}}\text{, by \cref{eq:decompose_tau}} \\
    &= U_\tau V_0 \ket{i}\text{, by reconstruction property of }\Xc,\text{ \cref{eq:phi_i_reconstruct}}, 
\end{align}
and thus we have shown that 
\begin{align}
    V_\tau \Gamma = U_\tau V_0.\label{eq: CFT-dynamical-duality}
\end{align}
This equation tells us that encoding $A_LB_L$ into the boundary and then time evolving to time $\tau$ is the same as applying $\Gamma$ to time evolve in the bulk, and then encoding the resulting $\tilde{A}_L\tilde{B}_L$ into the CFT in the appropriate way.

The above reasoning assumes that CFT correlation functions exactly capture their counterparts in the bulk and emergent systems.
However, these properties are only approximately true. Thankfully the error remains controlled, as we discuss in \cref{subsec:error}.

\subsection{Mapping the NLQC inputs/outputs to the simulation }\label{subsec:bulk-simulation}

To have a sensible NLQC protocol, we must consider a finite memory simulation of the CFT such as described in \cref{subsec: formal-definition-simulation}.
In order to prove in the next subsection that the simulation acts as a resource to non-locally implement the bulk dynamics $\Gamma$, we need to show that it mimics the CFT feature that bulk dynamics are logically implemented by a local physical time evolution.
Mathematically, this is captured by the following equation, 
\begin{align}
    U'_\tau V_0'&= V_\tau' \Gamma \label{eq: simulated-dynamics-implement-computation},
\end{align}
where $V_0'$ and $V_\tau'$ are isometries mapping to $\mathcal{H}_{sim}$ from $\mathcal{H}_{A_LB_L}$ and $\mathcal{H}_{\tilde{A}_L\tilde{B}_L}$ respectively, and $U'_\tau$ is time evolution in the simulation.
This is essentially a version of \cref{eq: CFT-dynamical-duality} in which the isometries map to the simulation Hilbert space rather than the CFT.
Our strategy for demonstrating it will be very similar to the proof of \cref{eq: CFT-dynamical-duality}, and make use of analogous versions of \cref{eq:phi_i_reconstruct,eq:phi_j_reconstruct}.

We now specify the simulation parameters needed to obtain this result, which will highlight that the aspects of the CFT that we required the simulation to capture in \cref{sec: simulating-holographic-cfts} are sufficient for non-local computation.
Recall from \cref{sec: simulating-holographic-cfts} that the part of a QFT captured by a simulation is specified by $\mathcal{M},\delta,T,m,$ and $\{t_k\}$.
The duration of time $T$ for which the simulation is accurate will need to be from slightly before $t=0$ to slightly after $t=\tau$, and we will only need to probe these two times, i.e.\ $\{t_k\}=\{0,\tau\}$.
We let the error tolerance $\delta$ remain a free parameter.
The vector space of CFT operators we will need to simulate is 
\begin{align*}
    \mathcal{M}=\underset{i,j}{\spn}\{ \Xc^i,\tXc^j \}.
\end{align*}

Let $\Xs$ and $\tXs$ be the corresponding simulation operators and choose $m=4$.
We can use exactly the same arguments as in \cref{subsec:bulk-boundary} and \cref{apx: ecc-from-correlation-functions} by replacing $\Xc$ and $\tXc$ with $\Xs$ and $\tXs$ wherever they appear in correlation functions which we can do using \cref{eq: CFT-sim-correlation-functions-relationship} and using the fact that in those arguments at most four operators appear in correlation functions.
This gives that there exist isometries $V_0'$ and $V_\tau'$ mapping to $\mathcal{H}_{sim}$ from $\mathcal{H}_{A_LB_L}$ and $\mathcal{H}_{\tilde{A}_L\tilde{B}_L}$ respectively such that
\begin{align}
    \Xs V_0'&= V_0' \Xl\label{eq: sim-init-ecc}
    \\
    \tXs V_\tau'&= V_\tau'\tXl ,\label{eq: sim-fin-ecc}
\end{align}
and such that \cref{eq: simulated-dynamics-implement-computation} holds.

Recall that by our simulation requirements the simulation operators $\Xs,\tXs$ have support on the lattice sites in the region on which their original counterparts $\Xc$ and $\tXc$ had support.
This together with the above conditions imply that these maps are error-correcting codes and inherit the property that any logical operator localized to $A_L$ ($B_L$) can be reconstructed on $S_W$ ($S_E$), and likewise for $\tilde{A}_L$ ($\tilde{B}_L$) and $S_N$ ($S_S$).
The relationship between the various systems we consider is summarized in \cref{fig:dream-within-a-dream}.

\begin{figure}
    \centering
    \includegraphics[width=0.7\textwidth]{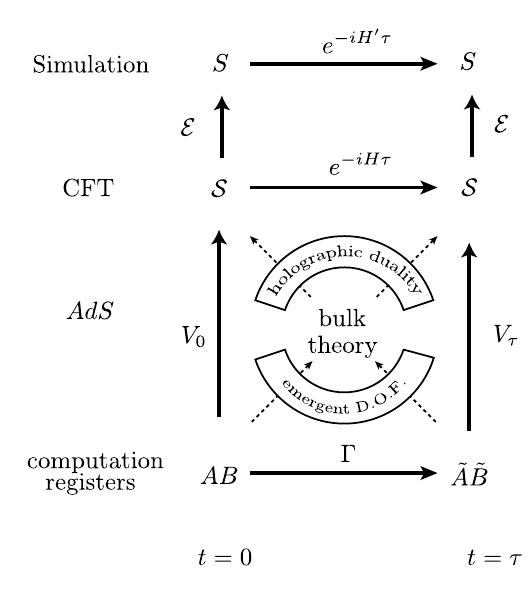}
    \caption{Relationship between the various systems we consider. }
    \label{fig:dream-within-a-dream}
\end{figure}

\subsection{Extracted protocol implements bulk dynamics}
\label{subsec:computation-validity}

Suppose Alice and Bob prepare the resource state $V'_0\ket{0}$ with $\west$ given to Alice and $\east$ to Bob, and are given an initial state $\ket{\psi}_{AB}$.
Since Alice can recover $A_L$ from $\west$, she can swap the information in $A$ into the encoded system $A_L$, and then throw away $A$.
This will constitute her encoding map $\mathcal{N}_{A\west\rightarrow\west}(\cdot)$.
Bob will do similarly with $\east$, $B$, and $B_L$.
They then hold 
\begin{align*}
    V_0'\ket{\psi}_{AB}.
\end{align*}
Because the simulation satisfies the no-superluminal signalling requirement, the simulated boundary time evolution operator will have a light-cone linear in $t$ with speed bounded above by $1$. So long as $\tau\leq 2\pi/8$ we can use the methods of \cref{sec: many-body-nlqc} to rewrite $U_\tau'=U_WU
_E
U_NU_S$. Alice applies $U_W$, Bob applies $U_E$, they exchange quarters appropriately, and apply $U_N$ and $U_S$. 
By \cref{eq: simulated-dynamics-implement-computation}, they then hold 
\begin{align*}
    V_\tau'\Gamma\ket{\psi}_{AB}.
\end{align*}
But $V_\tau'$ allows $\tilde{A}_L$ and $\tilde{B}_L$ to be decoded from $S_N$ and $S_S$ respectively. Thus Alice and Bob can locally reconstruct them, completing the protocol.

\subsection{Error tracking} \label{subsec:error}
We would like to upper bound the value of 
\begin{align*}
    \dnorm{\mathcal{R}_\tau\circ\mathcal{V}_{sim}\circ\mathcal{E}_0-\mathcal{U}_{L}},
\end{align*}
where
\begin{itemize}
    \item $\mathcal{U}_{L}(\rho_{AB}):=\Gamma\rho_{AB}\Gamma^\dagger$ is the channel implementing the bulk dynamics,
    \item $\mathcal{E}_0$ is the encoding action of Alice and Bob into the $t=0$ error-correcting code (a generalization of the \emph{isometric} encoding $ V'_0 $ that occurs in the exact setting described above),
    \item $\mathcal{V}_{sim}$ is the attempt of Alice and Bob to implement the simulation time evolution, and
    \item $\mathcal{R}_\tau$ is the decoding action of Alice and Bob from the $t=\tau$ error-correcting code,
\end{itemize}
Letting $\mathcal{C}_t$ be the encoding map at time $t$, the sources of error we must consider are
\begin{itemize}
    \item $ \epsilon_{enc}:=\dnorm{\mathcal{E}_0-\mathcal{C}_0}$, because Alice and Bob are restricted to encoding locally into $W$ and $E$, and the local reconstructions are only approximate,
    \item $ \epsilon_{rec}:=\dnorm{\mathcal{R}_\tau\circ\mathcal{C}_\tau - \mathds{1}_{AB}}$, because Alice and Bob must reconstruct locally fron $N$ and $S$,
    \item $\epsilon_{dyn}:=\dnorm{\mathcal{U}_{sim}\circ\mathcal{C}_0-\mathcal{C}_\tau\circ\mathcal{U}_L}$ with $\mathcal{U}_{sim}:=e^{-iH_{sim}\tau}\cdot e^{iH_{sim}\tau}$, because the dynamical duality is only approximate, and
    \item $\epsilon_{spread}:=\dnorm{\mathcal{V}_{sim}- \mathcal{U}_{sim}}$ because the simulation time evolution has only an approximate light-cone.
\end{itemize}

where $\mathcal{C}_t$ is an ideal encoding. 
With these definitions we can derive the bound

\begin{align}
    ||\mathcal{R}_\tau\circ\mathcal{V}_{sim}\circ\mathcal{E}_0-\mathcal{U}_{L}||_\diamond\leq \epsilon_{code}+\epsilon_{rec}+\epsilon_{spread}+\epsilon_{dyn}.
\end{align}

We calculate these errors in appendices \ref{appendix: err-from-aprox-corr-func} and \ref{appendix: err-from-lightcone}, and find that

\begin{align*}
    \dnorm{\mathcal{R}_\tau\circ\mathcal{V}_{sim}\circ\mathcal{E}_0-\mathcal{U}_{L}}\leq c_{CFT}\sqrt{G_N}+c_{sim}\sqrt{\delta}+c_{spread} a \exp(-b ( 2 \pi /8 - \tau))
\end{align*}
for some $O(1)$ numbers $c_{CFT}$, $c_{sim}$, and $c_{spread}$.
The idea is that the right hand side should be bounded above by some small constant $\epsilon$ as $n$ varies.
The first term decreases with $n$ and can easily be made smaller than e.g.\ $\epsilon/3$ by adjusting the starting value of $G_N$.
The simulation can be chosen with sufficiently high resolution so that $\delta<(\frac{\epsilon}{3a_{sim}})^2$, according to \cref{subsec: formal-definition-simulation}.
We found there that this requires a choice of $N_S$ that scales polynomially in $\frac{1}{\delta}$.
Finally, we fix $(\frac{2\pi}{8}-\tau)$ to be some small constant $\Delta \tau$, which means we require $b \geq \frac{1}{\Delta \tau} \log(\frac{3a_{spread}a}{\epsilon})$.
Because $a$ grows at most exponentially in $N_S$, this means that the error will be bounded if $b$ grows faster than polynomial in $N_S$, which is precisely what was required for the simulation to be efficient.

\subsection{Entanglement consumption } \label{subsec:entanglement}

We now give a rough estimate of the amount of entanglement between $\west$ and $\east$ that the simulation uses to perform this non-local computation.
From this we can learn about constraints on the bulk computation in the full AdS/CFT, as we discuss in \ref{subsec: implications-for-holography}. 

We expect that the entropies of simulation subregions scale similarly to the entropy of the corresponding regions in the CFT with the UV cutoff methods used in Refs.~\cite{Ryu_2006,Calabrese_2004,Holzhey_1994}, since the scaling of the entropies with $G_N$ generally does not depend on regularization, see for example \cite{Sorce-cutoff-covariant}.

The relevant entanglement is that between $\west$ and $\east$, since these are the resource systems given to Alice and Bob respectively. 
Suppose $\rho_{\west \east}$ is in a pure state. Then we can measure the entanglement via the von Neumman entropy $H(\west)$. 
In full AdS/CFT this is given by the Ryu-Takayanagi formula \cite{Ryu_2006,Hubeny_2007}, which in the specific context of considering one half of the hyperbolic disk, gives
\begin{align*}
    H(\west)= \frac{\mathrm{Area}(\gamma)}{4G_N}+O(1),
\end{align*}
with $\gamma$ a bisection of the disk. 
If we were using the complete geometry this quantity would be infinite.
However, we are working with an approximation of the CFT that amounts to introducing an IR cutoff in the bulk \cite{Ryu_2006}.
Suppose we decide to place the cutoff a distance $L$ from the bulk center.
Then, the area (length) of $\gamma$ is $2L$, and the entanglement between $\west$ and $\east$ is $\frac{2L}{4G_N}+O(1)$.

Such a bulk cutoff implies a boundary UV cutoff at a distance scaling like $e^{-L}$ \cite{Ryu_2006}.
Because the physics that we require the simulation to capture is concentrated in the centre of the bulk, our bulk cutoff $L$ need only be an $O(1)$ constant.
Increasing the cutoff further only allows the simulation to model smaller-scale properties of the CFT, corresponding to larger regions of the bulk space, which is irrelevant to the computation occurring in the center.

For at least some choices of unitary, NLQC protocols are known to require an increasing amount of entanglement to implement increasingly large input system sizes \cite{beigi-koenig}.
It may seem at first that our construction violates this lower bound, as the length of the cutoff is fixed at $O(1)$, independent of the number of input qubits $n$.
However, an increase in the number of qubits results in gravitational backreaction, which alters the geometry and ruins the validity of the argument that an NLQC protocol can be extracted.
The additional entropy then comes from increasing the size of the registers at each lattice site rather than increasing the number of sites.
This would make the entanglement used by AdS/CFT for the non-local computation at most $O\left(1/G_N\right)$.
This result differs only by an $O(1)$ factor with the one derived in Ref.~\cite{Alex-connect-wedge-theorem}, and thus many of the same conclusions can be reached.
While a more careful treatment will be necessary to determine the precise scaling of the entanglement with $1/G_N$, any polynomial scaling would yield stringent constraints on bulk computation, and imply that any unitary that can occur in the bulk is efficiently implementable in an NLQC protocol.
An exponential scaling would be surprising, since it would mean that a holographic CFT could effectively never be realized on a lattice, as it would require doubly exponentially many sites.

\section{Which computations are possible?} \label{sec:computations}
Our construction in this work describes an NLQC protocol implementing an isometry that uses the simulation of a holographic theory as a resource, with the isometry matching the bulk dynamics of that theory.
The entanglement cost $E$ of the protocol depends, according to the Ryu-Takayanagi formula \cite{Ryu_2006}, on the area of a certain region in units of Newton's constant, $G_N$.
In our setup we vary parameters of the holographic theory, such as $G_N$, while keeping the geometry fixed, meaning that the entanglement cost scales as $E \sim \frac{1}{G_N}$.
To assess the efficiency of our construction, then, we need to understand the constraints on $G_N$.
This also serves to help clarify which unitaries may be realized as the bulk dynamics of a holographic theory.

The specific setup is similar to that of previous work \cite{Alex-complexity,Alex-tasks}.
Consider the scenario depicted in \cref{fig: bulk-dynamics}. 
We refer to the domain of dependence of the smallest spacetime region containing the process $\Gamma$ as the \emph{interaction region}.
We consider a family of such processes, parameterized by the total number of input qubits $n$, and implementing some family of isometries $U(n)$.
We fix the interaction region geometry in proportion with the AdS radius, which is held constant as we vary $n$.
However, we vary $G_N$, which changes the Planck scale, and we also independently change the characteristic energy $E_c$\footnote{
Here we use $E_c$ to distinguish the characteristic energy of a component from $E$, the entanglement cost in terms of EPR pairs.
} of a component of the system (which bounds its size via the Compton wavelength).
In particular, as $n$ increases, we take these scales to be smaller relative to the AdS radius, i.e.\ $G_N\to 0$ and $E_c \to \infty$.

Any holographic bulk dynamics within the interaction region can be implemented via the NLQC protocol described in this work.
We argue that even in a rather constrained setting, the bulk dynamics can be quite general.
Specifically, our construction uses a circuit model to perform the computation, and we enforce fairly restrictive constraints on how such a computer fits within the interaction region. 
We do not even assume the existence of a universal quantum computer in the bulk; the entire construction can be tailored to the specific family of unitaries being considered.

In \cref{sss:gravity}, we review four constraints on what can occur in the interaction region.
First, the covariant entropy bound -- a bound on the density of information in the presence of gravity -- must be satisfied.
Second, the components (whether that be just the $n$ qubits themselves, or also some additional infrastructure) must fit within the volume of the interaction region.
Third, there is a time constraint that arises according to the time-complexity of $U(n)$.
Finally, the gravitational backreaction must not be too strong, or the geometry of the setup will be affected which may jeapordize the recoverability of outgoing information.

Previous work in similar contexts \cite{Alex-complexity} argued that these constraints imply $E\sim \frac{1}{G_N} \gtrsim \max(n,t_c(n))$ with $t_c(n)$ the time-complexity of $U(n)$.
Note that here we use $a \gtrsim b$ to denote $a = \Omega(b)$ in big-Omega notation.
We argue that appropriately accounting for the scaling of the characteristic component energy results in the stronger (although less rigorous) bound $E\sim \frac{1}{G_N} \gtrsim \max(s_c(n)^{\frac{D}{D-1}},s_c(n) t_c(n))$, with $D$ the spacetime dimension and $s_c(n)$ the spatial cost of performing a circuit for $U(n)$ (which must be at least $n$).

In \cref{subsec:assumption}, we consider the consequences of a specific assumption: that the above constraints are not only necessary but sufficient for the existence of a family of quantum computer states that can implement the family of circuits.
That is, provided a hypothetical circuit would not warp the geometry, satisfies the covariant entropy bound, and fits in both time and space, then it should be possible to build that circuit in the bulk.
This amounts to a statement about the computational power of bulk holographic theories.
With this assumption, we show that any polynomially complex family of unitaries can be performed non-locally with the entanglement cost scaling only polynomially in $n$.

\subsection{Constraints on bulk quantum computation}\label{sss:gravity}
As argued in \cite{Alex-complexity}, the covariant entropy bound can be used to argue that $n \leq \frac{\mathrm{Area}}{G_N}$.
Because the area of the interaction region is fixed, this immediately implies that $\frac{1}{G_N} \gtrsim n$.
This bound is essential for consistency with the known lower bounds on entanglement for NLQC from \cite{beigi-koenig}.
However, we can strengthen our results by enforcing even stricter (but less rigorous) bounds as follows.

The second constraint is that the qubits, and any additional components, must fit within the volume of the interaction region.
We denote the number of components required $s_c(n)$, which is effectively the space or memory cost of the computation $U(n)$.
For a typical quantum computer this would scale like the number of physical qubits, which is usually quasilinear in $n$, i.e.\ $s_c(n) \sim n \log^k n $ -- regardless, it must scale as at least $n$.
In general, the characteristic length $\ell$ of a component is lower bounded by its Compton wavelength, which scales inversely with its energy, i.e.\ $\ell \gtrsim \frac{1}{E_c}$.
For all $s_c(n)$ components to fit into the $n$-independent volume of the interaction region, we need $ \frac{s_c(n)}{E_c^{D-1}} \lesssim s_c(n) \ell^{D-1} \lesssim \mathrm{Vol} \sim O(1)$, i.e.\ $E_c \gtrsim s_c(n)^{\frac{1}{D-1}}$.

Thirdly, we have constraints arising from the computation time required for the family of unitaries $U(n)$ since, as argued in \cite{Jordan_2017,Lloyd_2000}, a limit on the information density and the speed of information propagation places a limit on computational speed.
If the time taken for a gate is $t_g$, the fixed physical time available in the interaction region is $t_p$, and the number of layers of gates needed is $t_c$, then we have $t_c t_g \leq t_p$.
For qubits that are spaced $\ell$ apart, the time required for a two-qubit gate must be at least proportional to $\ell$ \footnote{
The Margolus-Levitin theorem \cite{MargolusLevitin} shows that the time taken to evolve a component to an orthogonal state is at least $\frac{\pi}{2E_{avg}}$ with $E_{avg}$ the average energy of the component; however this assumes a Hamiltonian with zero ground state energy, and only applies to gates that orthogonally transform a state, both of which make it a weaker bound than the one we impose here.
}.
Thus $\frac{1}{E_c} \lesssim \ell \lesssim t_g \lesssim \frac{1}{t_c}$, so $t_c \lesssim E_c$.
Thus, if the time-complexity scales as $t_c(n)$, then we have $E_c \gtrsim t_c(n) $, which may be a stronger condition than the one above; generally the space and time bounds combine to give $E_c \gtrsim \max(s_c(n)^{\frac{1}{D-1}}, t_c(n) )$.

On its own, it appears one can simply use arbitrarily small (and thus high-energy) components to make the bounds from space and time restrictions easy to satisfy.
However, the higher the energy of the components, the greater the risk of gravitational backreaction that alters the geometry.
Consider the Einstein field equations,
\begin{align}
    G_{\mu\nu} = 8 \pi G_N T_{\mu \nu}.
\end{align}
The energy density tensor $T_{\mu \nu}$ scales linearly with the number of components $s_c(n)$, and with the characteristic energy $E_c$ of each component.
It also scales inversely with the volume of the interaction region, but this is held constant.
Thus we require the scaling $\frac{1}{G_N} \gtrsim s_c(n) E_c $ so that the gravitational backreaction is weak (relative to the AdS scale).

In summary, we have the following four bounds,
\begin{itemize}
    \item $1/G_N \gtrsim n$
    \item $E_c \gtrsim s_c(n)^{\frac{1}{(D-1)}}$
    \item $E_c \gtrsim t_c(n) $
    \item $1/G_N \gtrsim s_c(n) E_c$.
\end{itemize}
All of these can be satisfied by choosing $E_c \sim \max(s_c(n)^{\frac{1}{D-1}}, t_c(n) )$ and $\frac{1}{G_N} \sim E_c s_c(n)$.
Many of these bounds are perhaps overly restrictive -- for example, it may be possible to circumvent our spatial restrictions by having qubits occupy the same space e.g.\ in a bosonic field.

\subsection{Polynomial entanglement for polynomially complex unitaries} \label{subsec:assumption}
We claim that a family of NLQC protocols exists for $U(n)$ with entanglement scaling polynomially in its time and space complexities $t_c(n)$ and $s_c(n)$.
This follows from (i) the results of \cref{sec: NLQC-via-holographic-states}, that an NLQC protocol exists implementing the bulk dynamics of a state in an $AdS_3$/$CFT_2$ theory with entanglement scaling as $E\sim \frac{1}{G_N}$, and (ii) the following assumption.
\begin{assumption} \label{ass:bulkQC}
For a given family of unitaries $U(n)$, there is an interaction region in AdS$_3$ so that for each $n$, there is a holographic theory within which one can construct a circuit that implements $U(n)$ within that region, with $\frac{1}{G_N}$ scaling polynomially with the space- and time-complexity of $U(n)$.
\end{assumption}
Why should this assumption be valid?
Firstly, it is evidently compatible with all of the constraints described in \cref{sss:gravity}.
One way to show it would be to argue that the second through fourth bounds above can all be saturated.

One possible concern arises when saturating the space and time bounds, which is that it requires that components can be built at the length scale of the Compton wavelength; i.e.\ not only $\ell \gtrsim \frac{1}{E_c}$ but $\ell \sim \frac{1}{E_c}$.
For example, if the bulk physics included the Standard Model, and components are built out of atoms, this relationship would break once $\frac{1}{E_c}$ is smaller than the Bohr radius.
However, we are free to change parameters of the theory; by simply scaling all dimensionful parameters of the theory in proportion to $ E_c$ -- other than the AdS radius and $G_N$\footnote{Typically, this just means varying the string scale while keeping $G_N$ and the AdS radius fixed.} -- we can decrease the Bohr radius itself, and keep the internal structure of the components identical.

A similar potential obstacle is the bound on the gate time, $\ell \lesssim t_g$.
Aside from this restriction arising from causality, there could be some other larger time scale controlling the time needed to apply a gate, making it hard to saturate this bound. 
However, this gate time should still be controlled by the dimensionful parameters of the theory -- again, other than $G_N$ and the AdS radius.
Under the scaling of those dimensionful parameters with $\frac{1}{E_c}$, the gate time would still scale appropriately.
Essentially, the space and time bounds can be saturated so long as all time and distance scales relevant to the functioning of the quantum computer are independent of $G_N$ and the AdS radius.

The only missing piece is the existence of such components, that are capable of implementing arbitrary gates.
First, note that the AdS theory can be quite general; Ref.~\cite{HPPS} shows that essentially any bulk Lagrangian for a scalar field can be realized by choosing an appropriate boundary theory.
This provides evidence that the bulk theory is sufficiently expressive to have matter fields that can interact complexly enough to make arbitrary quantum circuits viable.
It is also known that certain theories are complex enough to facilitate arbitrary quantum computations -- including scalar $\phi^4$ theory with external potentials \cite{phi4}, the Bose-Hubbard model \cite{Bose-Hubbard}, and the Fermi-Hubbard model \cite{Fermi-Hubbard}.
Thus, although we cannot rigorously prove \cref{ass:bulkQC}, we hold this assumption to be physically well-motivated.

\section{A toy model}\label{sec: toy-model}

We now provide a toy model for a holographic code with dynamics that can perform Clifford gates non-locally. 

Consider the $7$-$3$ CSS code of Ref.~\cite{CSS-holographic-codes}, shown in \cref{fig:7-3 code}.
Because it is built from the Steane 7 qubit code, it has transversal Clifford gates.
Our toy model consists of two of these codes stacked on top of each other.
One begins slightly translated to the left, the other to the right, so that when the stack is cut in to two pieces along the vertical line, the person holding $\west$ is able to access the qubit marked in red on one of the codes, while the person holding $\east$ is able to access the red marked qubit on the other code, see \cref{fig:translated 7-3 code}.
Our plan will be to (i) apply a unitary operator that translates both codes in the stack so that the red qubits align, (ii) perform a Clifford operation between them, and (iii) translate one to be recoverable on the top half, and one to be recoverable on the bottom.
The existence of unitaries that perform these translations requires the presence of some auxiliary systems, as shown in \cref{fig:traslated_ancillas}.
The translations spread information along the boundary by less than an eighth of its circumference, as shown in \cref{fig:boundary-change-from-translation}.
The Clifford operation does not spread information at all since it is transversal.

Thus we have all the ingredients neccesary to apply many-body NLQC to use this toy model as a resource state for performing Cliffords non-locally.
By adding more codes to the initial stack that are also translated to the left or the right we can generalized this to arbitrary $n$ qubit Cliffords, thus this protocol is efficient in the sense that the entanglement required scales linearly with the number of qubits.
Note that is already known how to perform $n-$qubit Clifford gates non-locally and efficiently, see for example \cite{Code-routing}.
The point here is to show that this can be done holographically.

We note also that this translation technique is essentially the same one used in Ref.~\cite{Osborne-Stiegemann}, but the addition of transversal Cliffords now allows for bulk interactions.

\begin{figure}
    \centering
    \includegraphics[width=0.6\textwidth]{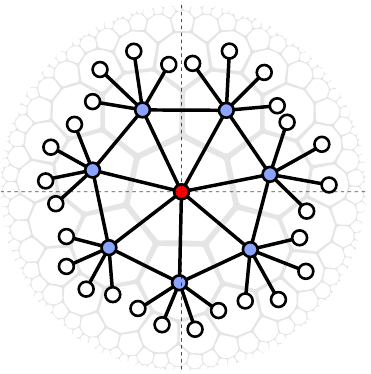}
    \caption{Holographic code constructed from Steane 7-3 code. The fundamental tensor it is built from admits transversal Cliffords, and thus so does the entire code.}
    \label{fig:7-3 code}
\end{figure}

\begin{figure}
    \centering
    \includegraphics[width=1.0\textwidth]{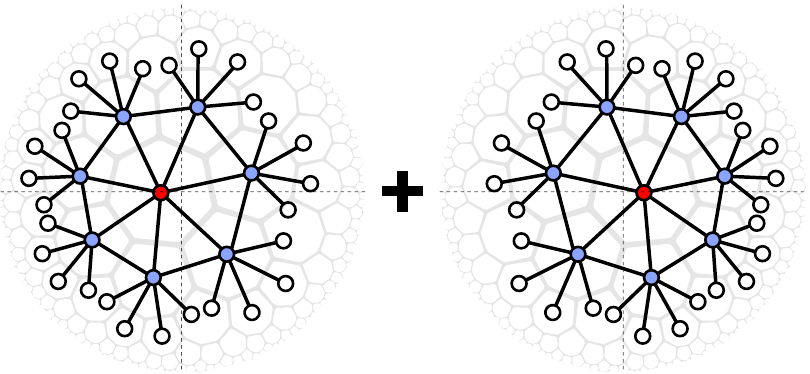}
    \caption{Our model consists of two such codes stacked on top of each other, one slightly translated to the left and one to the right. In this way when the stack is cut along the vertical the logical qubit indicated by the red dot can be reconstructed either on the west or east half depending on which code in the stack it's in.}
    \label{fig:translated 7-3 code}
\end{figure}

\begin{figure}
    \centering
    \includegraphics[width=0.5\textwidth]{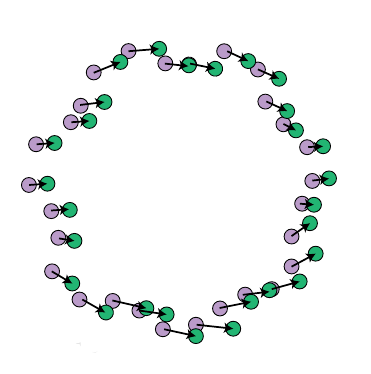}
    \caption{Code translation. To translate the code to the right, add ancillas (green dots) at the sites where the legs need to move, then swap the original physical registers (pink dots) appropriately. To translate back apply the exact same swap operation. Note that when the two codes in the stack are translated back to the center, their registers align, and thus a transversal gate can be applied. Another set of ancillas will need to be added to each of the codes to allow up and down translations. As more layers are added, the location of all dots converges towards the boundary.}
    \label{fig:traslated_ancillas}
\end{figure}

\begin{figure}
    \centering
    \includegraphics[width=1.0\textwidth]{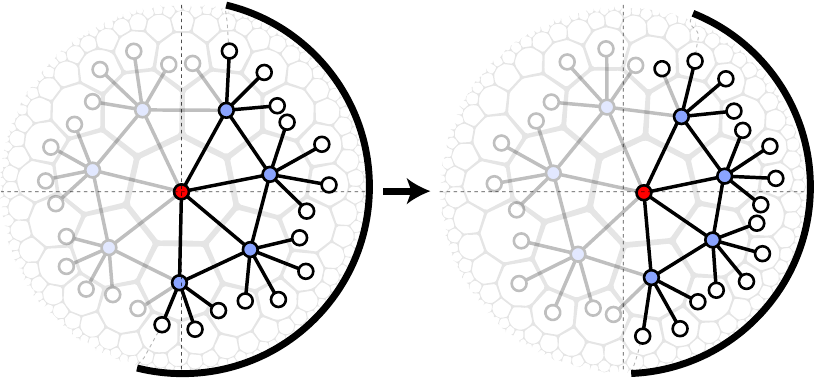}
    \caption{
    A right translation spreads information by less than $1/16$ of the boundary circumference. The spread of any vertical translation followed by a horizontal one will be less than $1/8$ as the many-body NLQC protocol requires. }
    \label{fig:boundary-change-from-translation}
\end{figure}


\section{Relation to previous work}
\label{sec: relation-to-previous-work}

Another demonstration that the CFT implements NLQC was given in Ref.~\cite{Alex-tasks}, as well as upper bounds on how much entanglement it consumes. However, derivation of these upper bounds relied on the assumption that certain boundary regions do not contribute. We repeat their demonstration below, and then argue that the contributions from these regions are in fact crucial.

Suppose the CFT is at Cauchy slice $\Sigma_c$ in a state $\ket{\psi({\Sigma_c})}$ in \cref{fig:alex-setup}, and that Alice and Bob%
\footnote{The language of Ref.~\cite{Alex-complexity} slightly differs from ours in that rather than Alice and Bob being two parties working together to implement the joint unitary, Alice has multiple agents trying to accomplish this. Bob instead is the one who gives Alice the original input systems at the input points and collects them at the output points.} 
can swap the information in finite memory systems they hold into CFT subsystems a points $c_0$ and $c_1$ respectively.
There should exist choices of $\ket{\psi({\Sigma_c})}$ such that as the state is evolved to Cauchy slice $\Sigma_r$, the bulk trajectory of these CFT subsystems is the one depicted in \cref{fig: bulk-task}, and that they undergo some joint unitary transformation in the center.
In the CFT picture, the information in $c_0$ ($c_1$) is processed in the $V_0$ ($V_1$) region, and then some information is sent to $W_0$ and some to $W_1$.
In $W_0$ ($W_1$), the information received from $V_0$ and $V_1$ is then processed so that at $r_0$ ($r_1$) the outgoing system is available.
This appears at first glance to have the form of a non-local computation protocol in which the resource systems $R_A$, $R_B$ contain entanglement drawn from $V_0$ and $V_1$ respectively.

\begin{figure}[t]
  \centering
    \includegraphics[width=0.6\textwidth]{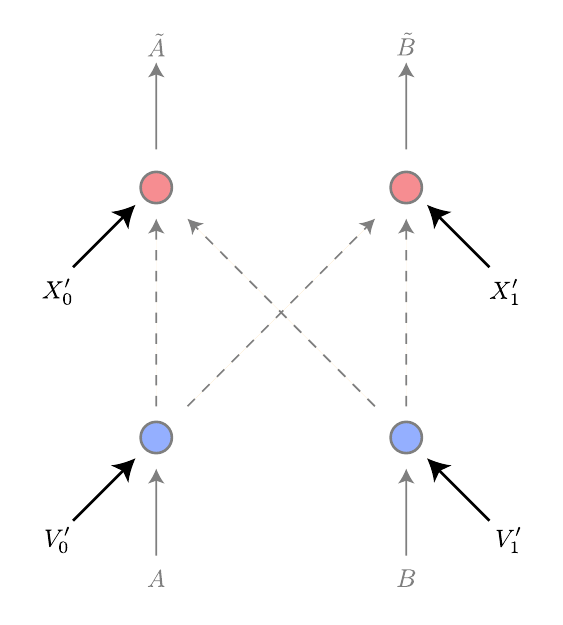}
    \caption{A non-local quantum computation variant in which some of the resource systems are not available until after the one round of communication.}
  \label{fig:NLQC-variant}
\end{figure}

However, additional resources in $X_0$ ($X_1$) may be involved in the post-processing stage that occurs in $W_0$ ($W_1$). 
Effectively, this is a variant model of NLQC in which some resources are initially unavailable until after the communication round, as depicted in \cref{fig:NLQC-variant}.
The assumption of Ref.~\cite{Alex-complexity} is that even in the variant NLQC protocol with post-communication resources, the mutual information between $V_0$ and $V_1$ still accurately represents the usefulness of the resource state for NLQC.
This setup has the nice property that the mutual information in the resource systems, i.e.\ $I(V_0:V_1)$ is then a finite quantity.
This then provides a limit to the amount of entanglement available to the protocol.
This may be used, for example, to place constraints on bulk dynamics from results known about NLQC.

\begin{figure}
    \centering
    \includegraphics[width=0.7\textwidth]{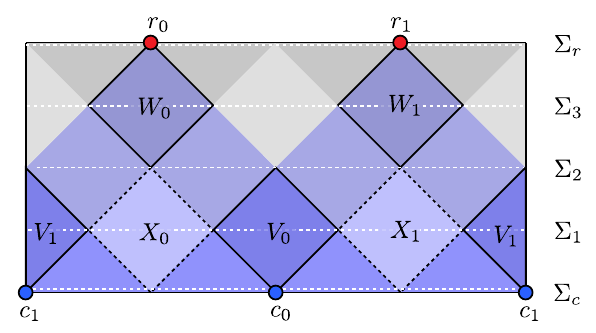}
    \caption{Aspects of the boundary relevant to the argument of Ref.~\cite{Alex-tasks} that AdS/CFT performs NLQC.}
    \label{fig:alex-setup}
\end{figure}

Three arguments are provided for this assumption:
The first is to view the assumption as a scientific hypothesis that has withstood several tests. In particular, this assumption has been used to arrive at several results relating mutual information and geometry in AdS/CFT which can be confirmed independently by gravity calculations \cite{Alex-connect-wedge-theorem,Alex-Branes,Alex-private-curves,AlexBeniJon}.

The second argument is that the region $P$ lies in the domain of dependence of the entanglement wedge\footnote{for this subsection only we use ``entanglement wedge'' to mean a spacetime region rather than a spacial region.} of $V_0V_1$. Thus all the data about the computation is contained in $V_0V_1$. 
However, direct extraction of an NLQC protocol requires making use of the locality of the boundary dynamics.
While the dynamics that implement the bulk computation can be reconstructed on the boundary in the region $V_0V_1$, this is generally a highly non-local reconstruction, in contrast to the local reconstruction with support on the full boundary.
When we argue below about the importance of the $X$ regions it will become clear that the data about the entanglement wedge of $W_0W_1$ is necessary in order to naturally produce an NLQC protocol in terms of $\ket{\Sigma_c}$ and the CFT Hamiltonian. 

The third argument provided in Ref.~\cite{Alex-complexity} relates to the general usefulness of post-communication resources in NLQC.
Ref.~\cite{Alex-connect-wedge-theorem} presents an example in which the variant NLQC protocol can use a resource state with no mutual information between $V_0$ and $V_1$ to accomplish an NLQC task for which the standard NLQC protocol is known to require such mutual information.
This appears to call into question the irrelevance of the additional post-communication resource systems $X_0$ and $X_1$, but the author points out that the state used in their example uses GHZ-type entanglement that is believed to be forbidden in holography, suggesting that in this restricted context the quantity $I(V_0:V_1)$ is still a good indicator of a resource state's usefulness.
However, we now argue that their example protocol (and in fact any protocol of this form) can be modified to remove the GHZ entanglement without adding mutual information between $V_0$ and $V_1$, once again calling the validity of the original assumption into question.

The resource state used in their example consists of four subsystems $V_0'V_1'X_0'X_1'$%
\footnote{the primed systems are finite-dimensional and represent the resources used by CFT across the regions $V_0,V_1,X_0,X_1$.},
and has no correlation between $V'_0$ and $V'_1$, i.e.\ $I(V'_0:V'_1)=0$.
Their specific protocol uses GHZ entanglement between $V'_0V'_1, X'_0,$ and $X'_1$.
However, an equivalent protocol arises by encrypting $X'_0$ with a one-time pad and adding a system containing the key to $V'_0$, so that Alice is able to decrypt $X'_0$ as soon as it is needed. 
No mutual information is added between $V'_0$ and $V'_1$, but the state is no longer equivalent to a GHZ state; this can be seen by tracing out $V'_0$ and noting that $X'_0$ is now maximally mixed, whereas a GHZ state would preserve classical correlations with $X'_1$.
Thus, this serves as an explicit counterexample to the assumption of Ref.~\cite{Alex-complexity} without relying on types of entanglement that are believed to not occur in holography.

Another way to see the importance of the $X$ regions is by considering explicit ways of evolving the system without them, and finding that an NLQC protocol can no longer be naturally extracted.
More explicitly, consider a channel which erases the $X$ regions from $\ket{\Sigma_1}$ and replaces them with some fixed state $\rho_X$, leaving the state
\begin{align*}
    \Tr_X{\ket{\Sigma_1}\bra{\Sigma_1}}\otimes \rho_X.
\end{align*}
We consider two possible ways to time evolve this system.
The first is to apply the original CFT Hamiltonian time evolution on all $V_0V_1X$.
The second is to remove terms which couple $V_0$ and $V_1$ to the $X$ regions, and evolve using this new Hamiltonian.

If we evolve in the first way, we must take into consideration the $X$ systems.
In QFT, if two coupled spatial regions are in a product state (i.e.\ unentangled), the energy-momentum tensor diverges at their boundary.
As we evolve the state, this ``firewall'' radically changes the course of events.
It is now unclear if a semiclassical geometry is even preserved outside the domain of dependence of $V_0V_1$.
If it is, the energy divergence manifests as a shock wave at the boundary of this region that causes pathological behaviour.
Although the computation can still occur, shortly afterwards the outward trajectories will be in the causal future of these pathologies, which will almost certainly cause sufficient noise to prevent the outgoing systems from being recovered from the $W$ regions.
Thus the system no longer shows signs of performing the NLQC.

If we evolve in the second way, then the $V_0$ and $V_1$ systems never interact and thus no bulk computation even occurs.
It could be argued that this is analogous to a non-traversable wormhole, and that the computation still occurs within the wormhole.
Even in this case, getting the systems out of the wormhole requires a direct interaction between $V_0$ and $V_1$ which itself needs to be implemented non-locally somehow.
Even if that is done, the absence of the $X$ regions means that when the particles come out of the wormhole they go back to near $c_0$ and $c_1$, and it does not appear that these new trajectories have a bulk/boundary causal structure discrepancy.
Thus again the system shows no signs of non-local computation.

In contrast, as we have seen, when the $X$ regions are included, an NLQC protocol naturally arises in terms of just the boundary state and Hamiltonian because of the guarantee that the outgoing systems reach the $W$ regions.

Finally, one might wonder whether if it is possible to show that for some unitary there must be entanglement present between $V_0'$ and $V_1'$ systems even in the modified NLQC task, as this would make possible the NLQC based proof technique suggested in \cite{Alex-connect-wedge-theorem} of certain relationships between bulk and boundary quantities.
However, in \cref{apx: entanglement-confined-to-X-regions} we give a proof that this is not so, i.e.\ for any unitary it is possible to perform the modified task with no entanglement between $V_0'$ and $V_1'$.

\section{Discussion }\label{sec: discussion}

AdS/CFT displays a discrepancy between the causal structure of the bulk and boundary. Specifically, it is possible for two particles coming in from the boundary to meet in the bulk, and then return to the boundary, while the same is impossible if they remain localized on the boundary (\cref{fig:causal-structure-discrepancy}). This suggests that the boundary performs non-local computation in some sense. However to rigorously prove that AdS/CFT performs the task of non-local quantum computation it is necessary to demonstrate operationally how the boundary state can be used as a resource for the task. Previous work has only done this non-constructively, and by assuming that large portions of the boundary can be ignored. 

In this paper we closed this gap by giving an explicit protocol in terms of a finite memory simulation of the CFT state. We started by showing how any one-dimensional local quantum system on a circular lattice implements some type of NLQC via the mechanism of many-body NLQC. We then argued that a finite-dimensional approximation of the CFT should exist in the form of a quantum simulation preserving the geometric properties of the CFT. Finally, applying many-body NLQC to the simulation gave a protocol which non-locally implements the local bulk computation. To make this concrete, we gave a toy model which can non-locally implement Clifford gates.



\subsection{Implications for Quantum Gravity }\label{subsec: implications-for-holography}

One application of this work to quantum gravity is to give an avenue for upper bounding the complexity attainable in the bulk for a specific geometric configuration, as discussed in \cite{Alex-connect-wedge-theorem,Alex-complexity}.
It has been conjectured in \cite{Alex-complexity} that there is a general lower bound on the entanglement required for a generic unitary family in terms of its complexity.
If this were the case, then we would be able to use arguments similar to those in Ref.~\cite{Alex-connect-wedge-theorem,Alex-complexity} to show an upper bound on complexity in the bulk.
That argument proceeds as follows.
Suppose the bound takes the form $E>f(C_U)$, with $E$ the entanglement required and $C_U$ the complexity of the bulk unitary $U$.
Take a finite memory simulation of the CFT with entanglement $E$, which can be used as a resource to perform $U$ non-locally.
This imposes an upper bound on complexity given by $C_U<f^{-1}(E)$ for implemented unitaries according to the hypothesized bound.
Of course, this only applies to those unitaries $U$ that are implemented in the bulk in an interaction region as in \cref{sec:computations}, but using plausible physical arguments we can extend such a bound to arbitrary regions.

Consider any unitary implemented in some sub-AdS scale spacetime region $R'$ in the holographic bulk.
Then we can construct a larger sub-AdS spacetime region $R\supset R'$ that has the properties required for the results of \cref{sec: NLQC-via-holographic-states} to apply.
We place a machine in $R\setminus R'$ that receives the input systems, feeds them into the computation in $R'$, then receives and routes them so that they leave $R$ in a direction perpendicular to their incoming trajectory.
Thus, $R$ is now an interaction region and so the results of \cref{sec: NLQC-via-holographic-states} apply, meaning that the NLQC protocol can be performed.
If this entanglement $E$ is not greater than the lower bound $f(C_U)$, we would arrive at a contradiction -- demonstrating that the unitary was too complex.
We can in fact obtain multiple constraints on the unitary through various partitions on the $R'$ input into the $A$ and $B$ inputs of $R$.  

\subsection{Implications for Quantum Information}

\subsubsection{Better NLQC protocols}

We argued in \cref{subsec:assumption} that if one accepts \cref{ass:bulkQC}, then all polynomially-complex unitaries can be performed non-locally with a polynomial amount of entanglement.
This would have drastic consequences for position-based cryptography.

The verifier would like to choose a computation which is too resource intensive for a dishonest prover using NLQC to perform. 
However, if the above argument holds true, then only high complexity unitaries would fulfill this requirement.
But such unitaries are also unfeasible for an honest prover to implement.
However, it may be possible to work around this by making assumptions on the limit of computational speed available to the prover. 

\subsubsection{Practical applications for holographic codes}
A number of toy models of holography have been studied that are built out of tensor networks, known as holographic codes \cite{HaPPY,CSS-holographic-codes,Gauging-the-bulk,bidirectional,spaghetti}.
Our method for converting approximate holographic states to resources for non-local computation works equally well for holographic codes if they exhibit dynamics in the following sense.
Let $V:\mathcal{H}_{bulk}\rightarrow \mathcal{H}_{boundary}$ be the encoding isometry of a holographic code.
We will say the code has dynamics if there exists local boundary time evolution that implements local bulk time evolution $U_{boundary}(t)V= V U_{bulk}(t)$ (we assume the code space does not change with time for simplicity).
In holographic codes both the physical and logical Hilbert spaces have an associated locality structure, in this case a closed 1-d lattice for the boundary and a 2-d hyperbolic lattice for the bulk.
Degrees of freedom live on the vertices and edges of the lattices, and time evolution is local with respect to them in the sense that the support of local operators conjugated by the time evolution can increase by a graph distance increasing with $t$\footnote{%
Note that the ``graph speed'' with which information can spread is constant on the boundary, i.e.\ the graph distance is proportional to $t$, but increases radially outwards in the bulk.
This is so because the speed with which the operators grow is the speed of light, and light travels at constant speed with respect to the cylindrical coordinates.}.
 A toy holographic code with these properties can have its bulk computation extracted via many-body NLQC. 

Thus far models have either lacked one of these features, e.g.\ superluminal signalling in Ref.~\cite{Kohler_2019}, or have trivial or non-interactive dynamics, e.g.\ \cite{Osborne-Stiegemann} and \cite{Gauging-the-bulk} respectively.
The toy model we introduced allows only for Clifford gates, and thus it remains to find codes with more intricate dynamics. 

This gives a potential practical application for holographic codes, since if one can be found which captures these features of AdS/CFT and has sufficiently intricate bulk dynamics, then by similar arguments we still expect to find a toy bulk state with a programmable quantum computer that can implement any unitary of bounded complexity. Even without such a state, the bulk dynamics may still yield a computation that current NLQC protocols cannot perform efficiently. 

\subsubsection{Suboptimality of Bell pairs}
It is interesting to note that the entropy $S$ between the two halves of the resource state, that is between two halves of a holographic CFT with a cutoff, does not agree with its (smooth) max Renyi entropy, implying non-maximal entanglement between the halves.
The latter goes like $S+O(\sqrt{S}$) \cite{Bao:2018pvs}.
On its own this fact is not remarkable, but there are several instances in quantum information in which such non-maximally entangled states perform better than maximally entangled states.
This cannot happen in a scenario in which both local operations and classical communication is allowed, as these operations are sufficient to convert $n$ maximally entangled Bell states into any bipartite state in which both sides hold $n$ qubits \cite{Nielsen-majorization}.
However, when communication is restricted, such as in non-local computation, then the possibility arises that the ideal resource is non-maximally entangled.
Two relevant examples are port-based teleportation \cite{port-teleportation} and teleportation by size \cite{quantum-gravity-in-the-lab}.
The former is the foundation of the best known general NLQC protocols, and its optimal resource state is not maximally entangled.
The latter is a toy model of information travelling through a wormhole.
There maximum entanglement corresponds to infinite temperature, where far less information can travel through.
Another, slightly more removed example is the embezzling state \cite{embezzling}, which allows one to coherently change the number of Bell pairs shared using only local operations. This again is a state with a non flat spectrum. 

Finally, it turns out that the set of local observables in quantum field theory is described not by the set of operators acting on a Hilbert space of finite or countably infinite dimension, but by what is known as a type $III_1$ algebra \cite{witten-entanglement-in-QFT}. These can be understood via the Araki-Woods construction, by considering sequences of operators that converge when acting on a state built from an infinite chain of non-maximally entangled qubit pairs. Simulating the CFT locally via the method of Ref.~\cite{Osborne-Stottmeister} amounts to truncating this chain at some value and computing only expectation values which are insensitive to the truncation. It seems possible that the dynamical duality present in AdS/CFT is reliant on this particular emergent algebraic structure, and since this requires non-maximally entangled states to simulate, and the simulation is thought to allow for new efficient NLQC protocols, perhaps we can find new protocols by studying such states more generally.
Thus a more general understanding of additional power granted by non-maximally entangled states in the context of local operations with limited classical communication may be useful in finding more efficient NLQC protocols or possibly holographic codes with dynamics.

\subsection{Future directions}

There are a number of natural future directions proceeding from this program.
The most promising is further study into what non-Clifford gates can be implemented transversally using multiple code blocks as in \cref{sec: toy-model}.
Codes with this property, as well as a threshold secret sharing scheme property, would immediately yield new classes of non-local computation protocols.

It would be interesting to explore whether there are tasks other than NLQC which are naturally associated with holography in higher dimensions, determine whether they have cryptographic interpretations, and to construct a general theory of such tasks. 

It would be good to know also whether in QFT simulations with no superluminal signalling it is necessarily possible to locally change the Cauchy slice, in which case the latter requirement would be strictly weaker, and should be used instead in the definition of simulation in order to determine the minimal properties of a simulation to extract the protocol.

It would be interesting also to find an expression for the implemented (pseudo) bulk dynamics $V_{AB\rightarrow\tilde{A}\tilde{B}}$
derived from many-body NLQC in terms of correlation functions of various local $(1+1)$-dimensional systems. Then if we can find a system whose correlation functions match those of a holographic system \cite{Maldacena-Duffin}, we can check to see if we can place a quantum computer in the bulk. This could be done, for example if some of the bulk dynamics could be described as a $\phi^4$ theory, by the methods of Ref.~\cite{Jordan-preskill}, where they describe how to build a universal quantum computer in such a theory given access to a potential one could control. It would remain then to figure out how to implement this potential using states of the theory itself.

Although we focused on a $(2+1)$-dimensional bulk, we note that a similar causal discrepancy exists in $3+1$ bulk dimensions.
In this case, a spatial slice of the boundary geometry is a $2$-sphere, and the four points ($c_0$, $c_1$, $r_0$ and $r_1$) can be placed maximally far apart in a tetrahedral configuration.
One can show that there is no point on the sphere $x$ with the property $\max(d(x,c_0),d(x,c_1)) + \max(d(x,r_0),d(x,r_1)) \leq \pi$, with $d(\cdot,\cdot)$ the arc distance along the sphere; this implies an analogous causal discrepancy as above.
However, it is unclear whether the non-local boundary process takes the specific form of an NLQC protocol.

It may also be possible to write down a direct simulation of a holographic CFT.
A candidate simulation ground state would be a tensor network state which mimics the entanglement structure of a holographic CFT, such as those proposed by Ref.~\cite{Bao:2018pvs}.
It would be necessary then to search for an appropriate discrete Hamiltonian, and check whether bulk correlation functions match what we expect from full AdS/CFT. 

It is interesting to note that although holography has allowed the application of information-theoretic tools and results to the study of gravity, it is uncommon that the connection is applied to obtain new results in quantum information theory.
For this reason, we find especially intriguing the idea that holographic theories perform novel kinds of information processing that can be applied to the study of NLQC.

\acknowledgments
We thank Alex May and Patrick Hayden for many discussions on these ideas over several years, and Florian Speelman for helping us prove that entanglement can be relegated purely to the $X$ regions.
We also thank Jon Sorce, Alexander Stottmeister, Aditya Cowsik and Raghu Mahajan for useful discussions. 

\appendix

\section{Simulation error correction properties from correlation functions}\label{apx: ecc-from-correlation-functions}

We will now show that the simulation inherits the error correcting properties of the CFT.
In particular we prove equations \cref{eq: sim-init-ecc} and \cref{eq: sim-fin-ecc}.
 In this appendix we treat the exact case, while in later appendices we take error into account. 

Suppose that at some fixed time $t$ there is an $N$-dimensional emergent bulk logical system $L$ contained in the entanglement wedge of a boundary physical region $R$ of the CFT. Suppose the system is of sufficiently low energy to avoid back reaction (see \cref{sec:computations}), and in particular, we take this to mean that the JLMS entropy condition \cite{Jafferis_2016} holds up to order $G_N$. The results of \cite{Jordan_2017} then tell us that for the bulk to boundary encoding map\footnote{$\mathcal{D}(\mathcal{H})$ denotes the set of density matrices on a Hilbert space $\mathcal{H}$.}
\begin{align*}
    \mathcal{N}:\mathcal{D}(\mathcal{H}_L)\rightarrow \mathcal{D}(\mathcal{H}_{CFT})
\end{align*}
there is a decoding map $\mathcal{R}$ that acts on the region $R$, which can approximately recover $L$ in the sense that for all $\rho_L\in\mathcal{D}(\mathcal{H}_L)$ we have that
\begin{align*}
    ||\rho_L-\mathcal{R}\circ\mathcal{N}(\rho_L)||_1\leq O(\sqrt{G_N}).
\end{align*}
Furthermore, letting $\mathcal{R}^\dagger$ be the adjoint of $\mathcal{R}$, we have that correlation functions are approximately preserved.
Specifically, for any set of logical operators $\{O^i_L\}_{i\in\{1,...,m\}}$, and letting $U_t=e^{-iH_Pt}$ , we have that
\begin{align}\label{eq: approx-corr-functions}
    | \ev{O_{L}^1 \cdots O_{L}^m } - \bra{\zc} U_t^\dagger \mathcal{R}^\dagger(O^1_{L}) \cdots \mathcal{R}^\dagger(O^m_{L}) U_t \ket{\zc}| \leq \sqrt{G_N}\times \mathrm{poly}(m)\times\opnorm{O_{L}^1} \cdots \opnorm{O_{L}^m},
\end{align}
where the first correlation function is taken with respect to some fixed logical state $\ket{0}_L$.
A very similar equation also arises when connecting the logical system to the simulation.
For this reason, we put it into a more abstract form.
For any $O_P^i$ in the image of $\mathcal{R}^\dagger$, i.e.\ $O_P^i = \mathcal{R}^\dagger(O_L^i)$, we define the map $\alpha_t'(O_P^i) = O_L^i$, which can in general be replaced by a linear and hermiticity-preserving map from a subspace of physical operators to logical operators.
Similarly $U_t$ is replaced by some general time evolution on the physical space.
Finally, we express the coefficient  on the right hand side as $\eta = \sqrt{G_N}\times\mathrm{poly}(m)$.

Then for any $\{O^i_P\}$ in the domain of $\alpha_t'$, defining the Heisenberg operators $O^i_P(t):=U_t^\dagger O_P^i U_t$, we have that
\begin{align}\label{eq: approx-corr-functions}
    | \ev{\alpha'_t(O_{P}^1) \cdots \alpha'_t(O_{P}^m) } - \ev{O_{P}^1(t) \cdots O_{P}^m(t) }| \leq \eta||\alpha'_t(O_{P}^1)\| \cdots \| \alpha'_t(O_{P}^m) ||.
\end{align}
We can now describe how an equivalent result arises in relating the simulation to the logical system directly.
We use a simulation so that the set of simulated operators is $\mathcal{M} = \mathcal{R}^\dagger(L(\mathcal{H}_L))$.
In that case, the condition from \cref{sec: simulating-holographic-cfts} that there exist a map $\alpha:\mathcal{H}_{sim} \to \mathcal{H}_{CFT}$ and a set of operators $\mathcal{M}'$ such that $\alpha(\mathcal{M}') = \mathcal{M}$ means that for every operator $M'\in \mathcal{M}'$, we can assign some corresponding operator $O_L^i \in L(\mathcal{H}_L)$ such that $\alpha(M') = \mathcal{R}^\dagger (O_L^i)$.
We define this assignment as $\alpha'_t(M') = O_L$.
Here $\alpha'_t$ can be chosen to be linear and hermiticity-preserving.
Then in \cref{eq: approx-corr-functions} we obtain $O^i_P(t) \defi U_t^{\prime \dagger}  M' U'_t $, with $U'_t$ the simulation time evolution.
This equation is valid using \cref{eq: CFT-sim-correlation-functions-relationship} by setting $\eta$ to now scale as $\sqrt{G_N}\mathrm{poly}(m) + \delta$.

\subsection{Exact case}
For the exact case, we set $\eta = 0$.
Now we wish to prove the following.
\begin{theorem}\label{thm: exact-isometry}
There is an isometry $V:\mathcal{H}_{L} \to \mathcal{H}_P$ such that for any logical operators $O_L$ there is a corresponding physical operator $O_P$ such that $O_P V = V O_{L}$ and $[O_P, V V^\dagger] = 0$.
\end{theorem}
\begin{proof}
Consider a basis $\{\ket{i_L}\}_{i=1}^{N}$ for the system $\mathcal{H}_{L}$.
Define the generalized Pauli operators $X_L = \sum_j \ket{(j+1)_L}\bra{j_L}$ and $Z_L = \sum_j \exp(2 \pi i j/N) \ketbra{j_L}$.
Let some choice of corresponding boundary operators be denoted $X_P(t)$ and $Z_P(t)$ respectively, noting that their Schrodinger picture versions $X_P\defi U_t X_P(t) U_t^\dagger $ and $Z_P\defi U_t Z_P(t) U_t^\dagger $ are localized to the region $P$.
Define $\ket{i} = X_P(t)^i \ket{0}$, and the isometry $V = U_t \sum_i \ketbra{i}{i_L}$.
This is indeed an isometry, because it maps an orthonormal basis to an orthonormal basis:
\begin{align}
    \braket{i}{j} &= \bra{0} X_P(t)^{i\dagger} X_P(t)^j \ket{0} \\
    &= \bra{0_L} X_L^{i\dagger} X_L^j \ket{0_L} \\
    &= \delta_{ij}
\end{align}
We now show that for any $i_L$, $a$, and $b$,
\begin{align}
    X_P^a Z_P^b V \ket{i_L} = V X_L^a Z_L^b \ket{i_L}, \label{eq:physicalpaulis}
\end{align}
from which the first claim follows by linearity.
Note that both sides of this equation are normalized, as the operators are each isometric; thus it suffices to show that the inner product of one side with the other gives one.
\begin{align}
     \bra{i_L}Z_L^{\dagger b} X_L^{\dagger a}  V^\dagger X_P^a Z_P^b V \ket{i_L} &=\sum_j \bra{i_L}Z_L^{\dagger b} X_L^{\dagger a}  \ketbra{j_L}{j} U_t^\dagger X_P^a Z_P^b U_t \ket{i} \\
     &=\sum_j \bra{i_L}Z_L^{\dagger b} X_L^{\dagger a}  \ketbra{j_L}{0}X_P(t)^{\dagger j} X_P(t)^a Z_P(t)^b X_P(t)^i \ket{0} \\
     &=\sum_j \bra{i_L}Z_L^{\dagger b} X_L^{\dagger a}  \ketbra{j_L}{0_L}X_L^{\dagger j} X_L^a Z_L^b X_L^i \ket{0_L} \\
     &=\sum_j |\bra{i_L}Z_L^{\dagger b} X_L^{\dagger a}  \ket{j_L} |^2 \\
     &= \sum_j \delta_{i+a,j} \\
     &= 1.
\end{align}
The codespace-preserving property $[VV^\dagger, X_P^a Z_P^b] = 0$ is now easy to show, using
\begin{align}
    X_P^a Z_P^b VV^\dagger &= X_P^a Z_P^b V(\sum_i \ketbra{i_L} ) V^\dagger \\
    &= V X_L^a Z_L^b (\sum_i \ketbra{i_L} ) V^\dagger\text{, using \cref{eq:physicalpaulis}} \\
    &= V (\sum_i \ketbra{i_L} ) X_L^a Z_L^b  V^\dagger \text{, commuting past the identity}\\
    &= V V^\dagger X_P^a Z_P^b  
\end{align}
\end{proof}

The theorem above can be straightforwardly generalized to the case in which there are a number $N$ of emergent systems, $L_i$, each contained in the entanglement wedge of a separate boundary region $P_i$.
In this case the isometry $V:\otimes_i \mathcal{H}_{L_i}\rightarrow \otimes_i \mathcal{H}_{P_i}$ will have the property that for any $i$ and and operator $O_{L_i}$ acting on $L_i$ will have a corresponding operator $O_{P_i}$ localized to $P_i$ such that $VO_{L_i}=O_{P_i}V$ and $[VV^\dagger,O_{P_i}]=0$.
This makes it a subsystem code, with the logical information in $L_i$ robust against the erasure of $\bar{P}_i$.

\section{Error from approximate correlation functions}\label{appendix: err-from-aprox-corr-func}

We now consider \cref{eq: approx-corr-functions} for the case of small but non-zero $\eta$. 

\subsection{Approximate code}
Let $f(\ketbra{i_L}{0_L})$ be an operator acting on the $P$ region of the physical system such that $\alpha'_t(f(\ketbra{i_L}{0_L}))=\ketbra{i_L}{0_L}$.
Then we define $\ket{i}:=f(\ketbra{i_L}{0_L})\ket{0}$. We can then define $V_t$ in the exact case of \cref{thm: exact-isometry}, i.e.\ $V_t=\sum_i U_t\ket{i}\bra{i_L}$.
Now $V_t$ will no longer be an exact isometry, but we now show it approximates one. First we will need the following simple lemma.

\begin{lemma}\label{lemma: sum-in-corr-func}
 Consider a set of logical operators $\{O^{i,\mu_i}_L\}$ with $i\in\{1,...,n\}$ and $\mu_i\in\{1,...,m_i\}$ for some integer $m_i$, and boundary operators confined to $P$ $\{O^{i,\mu_i}_P\}$ which map to them respectively under $\alpha'_t$. Then
 \begin{align*}
      | \sum_{\mu_i}(\ev{O_{L}^{1,\mu_1} \cdots O_{L}^{n,\mu_n} } - \ev{O_{P}^{1,\mu_1}(t) \cdots O_{P}^{n,\mu_n}(t) })| \leq \eta||\sum_{\mu_i}O_{L}^{1,\mu_1} \cdots O_{L}^{n,\mu_n} ||
 \end{align*}
\end{lemma}
\begin{proof}
Let $\tilde{O}^i_L:=\sum_{\mu_i}O^{i,\mu_i}_L$ and $\tilde{O}^i_P:=\sum_{\mu_i}O^{i,\mu_i}_P$. Because $\alpha'_t$ is linear, we have that $\alpha'_t(\tilde{O}^i_P)=\tilde{O}^i_L$. The result then follows by the assumption \cref{eq: approx-corr-functions}.
\end{proof}

With this we are ready to prove

\begin{lemma}
 $\opnorm{V_t^\dagger V_t-\mathds{1}}=O(\eta)$
\end{lemma}

\begin{proof}
\begin{align*}
    \opnorm{V_t^\dagger V_t-\mathds{1}}&= \sum_{i,j}\braket{\psi}{i}\braket{i_L}{j_L}\braket{j}{\psi}-1
\end{align*}
for the $\ket{\psi}$ which maximizes this expression. Clearly $\ket{\psi}$ can be chosen to be in the span of $\{\ket{i}\}$. Writing $\ket{\psi}=\sum_i c^i_\psi \ket{i}$ we have
\begin{align*}
    \opnorm{V_t^\dagger V_t-\mathds{1}}&= \sum_{i,j,k,l}c^{k*}_\psi c^{l}_\psi\braket{k}{i}\braket{i_L}{j_L}\braket{j}{l}-1\\
    &=\sum_{i,k,l}c^{k*}_\psi c^{l}_\psi\braket{k}{i}\braket{i}{l}-1\\
    &=  \sum_{i,k,l}c^{k*}_\psi c^{l}_\psi\bra{0}f^\dagger_t(\ketbra{k_L}{0_L})f_t(\ketbra{i_L}{0_L})\ket{0}\braket{i}{l}-1\\
    &\leq \left|\sum_{i,k,l}c^{k*}_\psi c^{l}_\psi\braket{0_L}{0_L}\braket{k_L}{i_L}\braket{0_L}{0_L}\braket{i}{l}\right|
    +\eta 
    ||\sum_{i,k,l} c^{k*}_\psi \braket{i}{l} c^{l}_\psi \ket{0_L}\braket{k_L}{i_L}\bra{0_L}||
    -1\\
    &= \left|\sum_{i,l}c^{i*}_\psi c^{l}_\psi\braket{i}{l}\right|+
    \eta ||\sum_{i,l} c^{l}_\psi c^{i*}_\psi \braket{i}{l}  \ket{0_L}\bra{0_L}||-1\\
    &=\left|\sum_{i,l}c^{i*}_\psi c^{l}_\psi\braket{i}{l}\right|(1+\eta)-1
\end{align*}
But
\begin{align*}
    \left|\sum_{i,l}c^{i*}_\psi c^{l}_\psi\braket{i}{l}\right| &=\left|\sum_{i,l}c^{i*}_\psi c^{l}_\psi\bra{0}f^\dagger_t(\ketbra{i_L}{0_L})f_t(\ketbra{l_L}{0_L})\ket{0}\right|\\
    &\leq \left|\sum_{i,l}c^{i*}_\psi c^{l}_\psi\braket{0_L}{0_L}\braket{i_L}{l_L}\braket{0_L}{0_L}\right|
    +\eta 
    ||\sum_{i,l} c^{i*}_\psi c^{l}_\psi \ket{0_L}\braket{i_L}{l_L}\bra{0_L}||\\
    &=1+O(\eta).
\end{align*}
so the result follows.
\end{proof}

As the approximation improves, we are guaranteed that $V_t^\dagger V_t$ is positive semidefinite. Thus we can define the exact isometry $\tilde{V}_t:=V_t(V_t^\dagger V_t)^{-\frac{1}{2}}$. With the previous lemma we can prove a number of useful properties of $V_t$ and $\tilde{V}_t$:

\begin{align*}
    \opnorm{V_t^\dagger V_t}\leq \opnorm{V_t^\dagger V_t-\mathds{1}}+\opnorm{\mathds{1}}=1+O(\eta)\\
    \opnorm{V_t}=\sqrt{\opnorm{V_t^\dagger V_t}}=1+O(\eta)\\
    \opnorm{(V_t^\dagger V_t)^{-\frac{1}{2}}}=\opnorm{V_t^\dagger V_t}^{-\frac{1}{2}}=1+O(\eta)\\
    \opnorm{V_t-\tilde{V}_t}\leq \opnorm{V_t}\times\opnorm{1-(V_t^\dagger V_t)^{-\frac{1}{2}}}=O(\eta)
\end{align*}

the last of which guarantees that all the previous hold when we replace $V_t$ with $\tilde{V}_t$.

\begin{lemma}
 For any operator $O_L$ acting on the logical system, and a physical operator $O_P$ confined to $P$ such that $\alpha'_t(O_P)=O_L$, we have that
 \begin{align*}
     \opnorm{V_tO_L-O_PV_t}&=O(\sqrt{\eta})\opnorm{O_L}\text{, and}\\
     \opnorm{[O_P,V_t V^\dagger_t]}&=O(\sqrt{\eta})\opnorm{O_L}.
 \end{align*}
 Furthermore the same holds if we replace $V_t$ with $\tilde{V_t}$.
\end{lemma}

\begin{proof}
\begin{align*}
    \opnorm{V_t O_L-O_P V_t}^2&=\bra{\psi_L}(O_L^\dagger V_t^\dagger-V^\dagger_t O^\dagger_P)(V_t O_L-O_P V_t)\ket{\psi_L}
\end{align*}
for a $\ket{\psi_L}$ which maximizes the expression. The first term is
\begin{align*}
    \bra{\psi_L}O_L^\dagger V_t^\dagger V_t O_L\ket{\psi_L}= \bra{\psi_L}O_L^\dagger  O_L\ket{\psi_L}+O(\eta)\opnorm{O_L}^2.
\end{align*}
The last term is 
\begin{align*}
    \bra{\psi_L}V^\dagger_t O^\dagger_P O_P V_t \ket{\psi_L}&=\sum c^*_i c_j \bra{i}O^\dagger_P O_P\ket{j}\\
    &= \sum c^*_i c_j \bra{0}f^\dagger_t(\ketbra{i}{0})O^\dagger_P O_Pf_t(\ketbra{j}{0})\ket{0}\\
    &\leq \sum c^*_i c_j \bra{i_L}O^\dagger_L O_L\ket{j_L}+O(\eta)\opnorm{\sum c^*_i c_j \ket{0}\braket{i_L}{j_L}\bra{0}}\times
    \opnorm{O^\dagger_LO_L}\\
    &= \bra{\psi_L}O_L^\dagger O_L\ket{\psi_L}+O(\eta)\opnorm{O_L}^2.
\end{align*}
The negative of the third term is 
\begin{align*}
    \bra{\psi_L}V^\dagger_t O^\dagger_P V_t O_L \ket{\psi_L}
    &=\sum_{i,j}c_i^*d_j \bra{i}O_P^\dagger\ket{j},
\end{align*}
where $\sum_j d_j\ket{j_L}=O_L\ket{\psi_L}$.
\begin{align*}
    &= \sum_{i,j}c_i^*d_j \bra{0}f^\dagger_t(\ketbra{i}{0})O_P^\dagger f_t(\ketbra{j}{0})\ket{0}\\
    &\leq \bra{\psi_L}O_L^\dagger O_L\ket{\psi_L}+O(\eta)\opnorm{O_L}^2.
\end{align*}
putting these together we have that 
\begin{align*}
    \opnorm{V_t O_L-O_P V_t}^2=O(\eta)\opnorm{O_L}^2\\
    \rightarrow \opnorm{V_t O_L-O_P V_t}=O(\sqrt{\eta})\opnorm{O_L}.
\end{align*}
That $\opnorm{[O_P,V_t V^\dagger_t]}=O(\sqrt{\eta})\opnorm{O_L}$ follows easily from this expression and the hermicity preserving property of $\alpha'_t$. Furthermore, that $V_t$ can be replaced with $\tilde{V}_t$ in both of these expressions follows easily from the fact that $\opnorm{V_t-\tilde{V}_t}=O(\eta)$.
\end{proof}
\begin{lemma}\label{lemma: aux-approx-reconstruction}
 For any operator $O_{LA}$ acting on the logical system and an auxiliary system, and an operator $O_{PA}$ acting on the $P$ region of the physical system together with the same auxiliary system, such that $\alpha'_t\otimes I_A(O_{PA})=O_{LA}$, then
 \begin{align*}
     \opnorm{V_t\otimes\mathds{1}_AO_{LA}-O_{PA}V_t\otimes\mathds{1}_A}&=O(\sqrt{\eta})\text{, and}\\
     \opnorm{[O_P,V_t^\dagger V_t]}&=O(\sqrt{\eta}).
 \end{align*}
 Furthermore the same holds if we replace $V_t$ with $\tilde{V_t}$.
\end{lemma}

\begin{proof}
Let $O_{LA}=\sum_iO_L^i\otimes O_A^i$ be a Schmidt decomposition up to normalization. By linear, we must have $O_{LA}=\sum_iO_P^i\otimes O_A^i$ such that $\alpha'_t(O_P^i)=O_L^i$, and where the $O_P^i$ need not be orthogonal. We then see that
\begin{align*}
 &\opnorm{V_t\otimes\mathds{1}_AO_{LA}-O_{PA}V_t\otimes\mathds{1}_A}^2\\ 
 &=\bra{\psi_{LA}}(V_t\otimes\mathds{1}_AO_{LA}-O_{PA}V_t\otimes\mathds{1}_A)^\dagger (V_t\otimes\mathds{1}_AO_{LA}-O_{PA}V_t\otimes\mathds{1}_A)\ket{\psi}_{LA}
\end{align*}
for the $\ket{\psi_{LA}}$ which maximizes the expression. Letting $\ket{\psi_{LA}}=\sum_i c_i \ket{i_L}\ket{i_A}$, we then have that the last term is 
\begin{align*}
    &\sum_{i,j}c_i^*c_j\bra{i_L}\bra{i_A}V_t^\dagger O_{PA}^\dagger O_{PA}V_t\ket{j_L}\ket{j_A}\\
    &=\sum_{i,j,k,l}c_i^*c_j\bra{i_P} O_{P}^{k\dagger} O^l_{P}\ket{j_P}
    \bra{i_A} O_{A}^{k\dagger} O^l_{A}\ket{j_A}\\
    &\leq\bra{\psi_{LA}}O_{LA}^\dagger  O_{LA}\ket{\psi_{LA}} + O(\eta)\opnorm{\sum_{i,j,k,l}c_i^*c_j\ket{0_L}\bra{i_L}O_L^{k\dagger} O^l_L\ket{j_L}\bra{0_L}\bra{i_A}O_A^{k\dagger}O^l_A\ket{i_A}}\\
    &= \bra{\psi_{LA}}O_{LA}^\dagger  O_{LA}\ket{\psi_{LA}} + O(\eta)\opnorm{O_{LA}}^2.
\end{align*}
The other terms follow likewise, giving the result
\end{proof}

\begin{lemma}\label{lemma: unitary-reconstruction}
For any unitary operator $U_{LA}$ acting on the logical system and an auxiliary system, there is a unitary operator $U_{PA}$ acting on the $P$ region of the physical system together with the same auxiliary system, such that
\begin{align*}
    \opnorm{V_t\otimes\mathds{1}_AU_{LA}-U_{PA}V_t\otimes\mathds{1}_A}&=O(\sqrt{\eta}).
\end{align*}
\end{lemma}

\begin{proof}
Let $O_{PA}$ be an operator such that $\alpha'_t\otimes I_A(O_{PA})=U_{LA}$ 
Let $O_{PA}=\sum_i \lambda_i \ketbra{i'}{i}$ be a singular value decomposition. Define $U_{PA}:=\sum_i \tilde{\lambda}_i\ketbra{i'}{i}$ where
\begin{equation}
\tilde{\lambda}_i =
    \begin{cases}
        1 & \text{if } \lambda_i=0\\
        {\lambda_i}/{|\lambda_i|} & \text{else.}
    \end{cases}
\end{equation}
It suffices to show that
\begin{align*}
    \opnorm{(U_{PA}-O_{PA})V_t}=O(\sqrt{\eta}).
\end{align*}
First note that
\begin{align*}
    \opnorm{(U_{PA}-O_{PA})V_t}^2
    &=\max_\psi \bra{\psi}V_t^\dagger (\sum_i \ketbra{i}{i}|1-\left|\lambda_i|\right|^2)V_t\ket{\psi}
\end{align*}
We also have that
\begin{align*}
    \opnorm{(O_{PA}^\dagger O_{PA}-\mathds{1})V_t}=O(\sqrt{\eta})
\end{align*}
by using \cref{lemma: aux-approx-reconstruction} twice, and thus
\begin{align*}
    \opnorm{(O_{PA}^\dagger O_{PA}-\mathds{1})V_t}^2&=\opnorm{V_t^\dagger(O_{PA}^\dagger O_{PA}-\mathds{1})^\dagger(O_{PA}^\dagger O_{PA}-\mathds{1})V_t}\\
    &= \opnorm{V_t^\dagger (\sum_i \ketbra{i}{i}|1-\left|\lambda_i|^2\right|^2)V_t}\\
    &\geq \bra{\psi}V_t^\dagger (\sum_i \ketbra{i}{i}|1-\left|\lambda_i|^2\right|^2)V_t\ket{\psi}\\
    &\geq \bra{\psi}V_t^\dagger (\sum_i \ketbra{i}{i}|1-\left|\lambda_i|\right|^2)V_t\ket{\psi}\\
    &= \opnorm{(U_{PA}-O_{PA})V_t}^2
\end{align*}
 and thus 
\begin{align*}
    \opnorm{(U_{PA}-O_{PA})V_t}&\leq O(\sqrt{\eta})\\
    \rightarrow  \opnorm{(U_{PA}-O_{PA})V_t} 
    &= O(\sqrt{\eta}).
\end{align*}
\end{proof}

\subsection{Encoding error}
We would like to upper bound

\begin{align*}
    \dnorm{\mathcal{E}_0-\mathcal{C}_0},
\end{align*}

where $\mathcal{E}_0$ is the encoding action of Alice and Bob, and $\mathcal{C}_0$ is an encoding isometry. Specifically, 
$\mathcal{C}_t:=\tilde{V}_t\cdot\tilde{V}_t^\dagger$. The encoding action $\mathcal{E}_0$ is not simply the application of $\tilde{V}_0$, because Alice and Bob each hold only half of the circular lattice, namely $W$ and $E$ respectively, and must perform local encodings.
Formally, the encoding will be a CPTP map from registers $AB$ to $WE$, taking the form
\begin{align*}
    \mathcal{E}_0(\rho_{AB})=\Tr_{AB}\tilde{\Sigma}_{WA}\otimes\tilde{\Sigma}_{EB}(\rho^R_{WE}\otimes \rho_{AB}),
\end{align*}
where $\tilde{\Sigma}_{WA}$ and $\tilde{\Sigma}_{EB}$ are the channels version\footnote{By the channel version of a unitary $U$ we mean the CPTP map $U\cdot U^\dagger$.} of the promised unitary operators from \cref{lemma: unitary-reconstruction} which approximately implement $\Sigma_{AA_L}$ and $\Sigma_{BB_L}$, with $\Sigma$ the channel version of the swap operator, and $\rho^R_{WE}$ is defined as.
\begin{align*}
    \rho^R_{WE}=\frac{V_0\ketbra{00}{00}_{A_LB_L}V_0^\dagger}{\Tr[V_0\ketbra{00}{00}_{A_LB_L}V_0^\dagger]}.
\end{align*}
To calculate $\opnorm{\mathcal{E}_0-\mathcal{C}_0}$ it will be convenient to rewrite $\mathcal{C}_0$ as
\begin{align*}
    \mathcal{C}_0(\rho)=\Tr_{A_LB_L}[\mathcal{C}_0\Sigma_{A_LA}\otimes\Sigma_{B_LB}(\ketbra{00}{00}_{A_LB_L}\otimes\rho_{AB})]
\end{align*}
We then have
\begin{align*}
    \dnorm{\mathcal{E}_0-\mathcal{C}_0}&\leq
    \dnorm{\tilde{\Sigma}_{WA}\otimes\tilde{\Sigma}_{EB}(\rho^R_{WE}\otimes \cdot)-\mathcal{C}_0\circ\Sigma_{A_LA}\otimes\Sigma_{B_LB}(\ketbra{00}{00}_{A_LB_L}\otimes\cdot)}\\
    &\leq
    \dnorm{\tilde{\Sigma}_{WA}\otimes\tilde{\Sigma}_{EB}\circ\mathcal{C}_0/\Tr[\mathcal{C}_0(\ketbra{00}_{A_LB_L})]-\mathcal{C}_0\circ\Sigma_{A_LA}\otimes\Sigma_{B_LB}}.
\end{align*}
But \begin{align*}
    |\Tr[\mathcal{C}_0(\ketbra{00}_{A_LB_L})]|&= |\Tr[V_0^\dagger V_0\ketbra{00}_{A_LB_L}]|\\
    &=1+O(\eta).
\end{align*}
and thus the expression reduces to 
\begin{align*}
    \dnorm{\mathcal{E}_0-\mathcal{C}_0}&\leq \dnorm{\tilde{\Sigma}_{WA}\otimes\tilde{\Sigma}_{EB}\circ\mathcal{C}_0-\mathcal{C}_0\circ\Sigma_{A_LA}\otimes\Sigma_{B_LB}}+O(\eta).
\end{align*}

Using the readily verifiable fact that for two isometries $V$ and $W$ we have 
\begin{align}\label{eq: diamond-for-isometries}
    \dnorm{V\cdot V^\dagger-W\cdot W^\dagger}\leq 2\opnorm{V-W} 
\end{align} and \cref{lemma: unitary-reconstruction} we finally get that
\begin{align*}
    \dnorm{\mathcal{E}_0-\mathcal{C}_0}=O(\sqrt{\eta}).
\end{align*}

\subsection{Decoding error}
We would also like to upperbound
\begin{align*}
    \dnorm{\mathcal{R}_\tau\circ\mathcal{C}_\tau-I_{\tilde{A}\tilde{B}\rightarrow \tilde{A}\tilde{B}}}
\end{align*}

where $\mathcal{R}_\tau$ is the decoding action of Alice and Bob. Once again, the encoding action $\mathcal{R}_\tau$ is not simply the undoing of the isometry $\tilde{V}_\tau$, because Alice and Bob each hold only half of the circular lattice, this time $N$ and and $S$ respectively, and thus perform local decodings.
Formally, the decoding map will be a CPTP map from registers $NS$ to $AB$, taking the form, taking the from

\begin{align*}
    \mathcal{R}_\tau(\rho_{NS})=\Tr_{NS}\tilde{\Sigma}_{NA}\otimes\tilde{\Sigma}_{SB}(\rho_{NS}\otimes \ketbra{00}_{AB}).
\end{align*}

But since 
\begin{align*}
    \dnorm{\tilde{\Sigma}_{NA}\otimes\tilde{\Sigma}_{SB}\circ \mathcal{C}_0-\mathcal{C}_0\circ\Sigma_{A_LA}\otimes\Sigma_{B_LB}}=O(\sqrt{\eta})
\end{align*}

We find 
\begin{align*}
    \dnorm{\mathcal{R}_\tau\circ\mathcal{C}_\tau-I_{\tilde{A}\tilde{B}\rightarrow \tilde{A}\tilde{B}}}=O(\sqrt{\eta}).
\end{align*}

\subsection{Dynamical duality error}
We would like to bound
\begin{align*}
    \dnorm{\mathcal{U}_{sim}\circ\mathcal{C}_0-\mathcal{C}_\tau\circ\mathcal{U}_L}.
\end{align*}

Recalling that $\mathcal{U}_{sim}=U_{sim}\cdot U_{sim}^\dagger$ where $U_{sim}=e^{-iH_{sim}t}$, that $\mathcal{U}_L=\Gamma\cdot\Gamma^\dagger$, and \cref{eq: diamond-for-isometries}, we can upper bound the left hand side by

\begin{align*}
    LHS&\leq \opnorm{U_{sim}\tilde{V}_0-V_\tau\Gamma}\\
    &=2\max_{\ket{\psi_L}}\sqrt{\bra{\psi_L}(U_{sim}\tilde{V}_0-V_\tau\Gamma)^\dagger(U_{sim}\tilde{V}_0-V_\tau\Gamma)\ket{\psi_L}}.
\end{align*}
All but the cross terms are $1$, and it is sufficient to analyze one of the cross terms, namely

\begin{align*}
    \bra{\psi_L}\tilde{V}_0^\dagger U_{sim}V_\tau\Gamma\ket{\psi_L}
    &=\bra{0_L}(\ketbra{0_L}{\psi_L})\tilde{V}_0^\dagger U^\dagger_{sim}V_\tau(\Gamma\ketbra{\psi_L}{0_L}\Gamma^\dagger)\Gamma\ket{0_L}\\
    &= \bra{0_L}\tilde{V}_0^\dagger f_0^\dagger(\ketbra{0_L}{\psi_L}) U^\dagger_{sim} f_t(\Gamma\ketbra{\psi_L}{0_L}\Gamma^\dagger)V_\tau\Gamma\ket{0_L}+O(\sqrt{\eta})\\
    &= \bra{0_{sim}} f_0^\dagger(\ketbra{0_L}{\psi_L}) U^\dagger_{sim} f_t(\Gamma\ketbra{\psi_L}{0_L}\Gamma^\dagger)U_{sim}\ket{0_{sim}}+O(\sqrt{\eta})
\end{align*}
where $f_0(\ketbra{0_L}{\psi_L})$ and $f_t(\Gamma\ketbra{\psi_L}{0_L}\Gamma^\dagger)$ are simulation operators which implement $\ketbra{0_L}{\psi_L}$ on $\mathcal{H}_{A_LB_L}$ and $\Gamma\ketbra{\psi_L}{0_L}\Gamma^\dagger$ on $\mathcal{H}_{\tilde{A}_L\tilde{B}_L}$ respectively. But $U^\dagger_{sim} f_t(\Gamma\ketbra{\psi_L}{0_L}\Gamma^\dagger)U_{sim}$ is the Heisenberg evolved version of $f_t(\Gamma\ketbra{\psi_L}{0_L}\Gamma^\dagger)$ and thus by \cref{eq:logical-to-CFT-corr-err}\footnote{this equation is in terms of CFT correlators, but can be converted to simulation correlators by replacing $G_N$ with $G_N+\delta$, which we just denote as $\eta$ here} this can be written as 
\begin{align*}
    \braket{0_L} \braket{\psi_L} 
    \braket{0_L}+O(\sqrt{\eta})=1+O(\sqrt{\eta})
\end{align*}
and thus we get

\begin{align*}
    \dnorm{\mathcal{U}_{sim}\circ\mathcal{C}_0-\mathcal{C}_\tau\circ\mathcal{U}_L}\leq O(\sqrt{\eta}).
\end{align*}

\section{Error from imperfect light-cone}
\label{appendix: err-from-lightcone}
\subsection{Unitary decomposition}
\label{app:spread}
Here we show that in order to approximately derive the decomposition \cref{eq: unitary-chopping}, it suffices for the approximate spread to be small; more concretely, that the approximate light-cone condition from \cref{sec: simulating-holographic-cfts} holds.

Recall the definition
$
    K_\phi=U_{S'}\Sigma_\phi U^\dagger_{S'}
$.
Our goal is to replace these with unitary approximations $\tilde{K}_\phi$ that are localized to individual subsystems from the set $\{\east,\west,\north,\south\}$, such that $\| K_\phi - \tilde{K}_\phi\|$ is small.
Then we can show that implementing the $\tilde K_\phi$s is sufficient for approximately applying $U_S \otimes U_{S'}^\dagger$.


For any choice of $\phi$, there is some region $R\in\{\east,\west,\north,\south\}$ such that the distance $d(\phi,R)$ is at least $2\pi/8$.
Let us collect each such set together into products, e.g.\ $K_N = \prod_\phi K_\phi$ for all $\phi$ at least $2\pi/8$ away from $\south$, and similarly for $K_S$, $K_E$ and $K_W$.
Then using the approximate light-cone condition of the simulation, we can use the following property of the Haar measure on unitaries on $R$,
\begin{align}
    \int \! \mathrm{d} U_R \ U_R K_{\bar R} U_R^\dagger  = K^{\bar R}_{\bar R} \defi \tr_R (K_{\bar R}) \otimes \Id_R,
\end{align}
along with the fact that $\int \! \mathrm{d} U_R \ U_R U_R^\dagger = 1$, to conclude that
\begin{align}
    \| K^{\bar R}_{\bar R} - K_{\bar R} \| &= \|  \int \! \mathrm{d} U_R \ [ U_R, K_{\bar R}] U_R^\dagger  \| \\
    &\leq  \int \! \mathrm{d} U_R \ \| [ U_R, K_{\bar R}] \| \\
    &\leq  c_1 \exp(-c_2 ( 2 \pi /8 - t))
\end{align}
However, we need $\tilde{K}_{{R}}$ to be unitary, which $K_ {\bar{R}}^{\bar{R}}$ may not be.
To resolve this, we can just do the following:
\begin{align}
    \tilde{K}_R \defi K_{\bar{R}}^{\bar{R}} (K_{\bar{R}}^{\bar{R}} K _{\bar{R}}^{\dagger {\bar{R}}})^{-1/2},
\end{align}
which clearly does not change the asymptotic behaviour of the error in $\tilde{K}_R$ of approximating $K_R$.

Now, the ordering of the $\tilde K_R $ can be chosen such that all the east and west  components are implemented first, and then the north and south components, so that the unitary implemented by Alice and Bob is $ V\defi  \prod_{\phi \in [0,2\pi)}\Sigma_\phi\prod_{R \in \{N,S,E,W\}}\tilde{K}_R $.
Then finally, we obtain
\begin{align}
     \| V - U_S \otimes U_{S'}^\dagger \| &\leq \sum_{R \in \{N,S,E,W\}} c_1 \exp(-c_2 ( 2 \pi /8 - t)) \\
     &\leq 4 c_1 \exp(-c_2 ( 2 \pi /8 - t)) =: \epsilon_{spread}.
\end{align}

\section{Proof that entanglement can be confined to $X$ regions}
\label{apx: entanglement-confined-to-X-regions}

We now show that the modified NLQC task described in section \cref{sec: relation-to-previous-work} can be accomplished with no entanglement between the $V_0'$ and $V_1'$ systems. 

We make heavy use of two types of quantum teleportation: normal teleportation \cite{normal-teleportation} and port-based teleportation \cite{port-based-teleportation}.

\begin{primitive} Normal teleportation

\begin{enumerate}
    \item \textbf{Setup:} Alice holds an $n_A$-qubit system $A$ which, together with some reference system $R$ is in a state $\ket{\psi}_{AR}$. Alice and Bob each hold one of a maximally entangled pair of $2^{n_A}$-dimensional systems $A'B'$. 
    \item Alice performs some joint measurement on $A$ and $A'$, receiving the identity of an $n_A$-qubit Pauli operator $P$ drawn uniformly from all such operators.
    \item Bob's half of the maximally entangled state, together with the reference system,
    is then in the state $P\ket{\psi}_{B'R}$, where $P$ acts on the $B'$ subsystem.
\end{enumerate}
We say that the system $A$ has been \emph{normal-teleported}.
\end{primitive}

\begin{primitive} Port-based teleportation

\begin{enumerate}
    \item \textbf{Setup:} Alice holds an $n_A$-qubit system $A$ which, together with some reference system $R$ is in a state $\ket{\psi}_{AR}$. Alice and Bob each hold one half of $N$ maximally entangled pairs of $2^{n_A}$-dimensional systems $\{A'_xB'_x\}$. 
    \item Alice performs some joint measurement on $A$ and $\{A'_x\}$, receiving a random number $x^*\in [N]$\footnote{We use the notation $[N]$ to denote the set $\{1,\ldots,N\}.$}.
    \item Bob's half of the $x^*$th entangled system, together with the reference system, is then close to the state $\ket{\psi}_{B_{x^*}R}$, with error going to zero as $N$ goes to infinity.
\end{enumerate}
We say that the system $A$ has been \emph{port-teleported}.

\end{primitive}

Consider three parties, Alice, Bob, and Charlie. Alice and Charlie share as many Bell pairs as needed for what follows, and similarly do Bob and Charlie. Alice and Bob share no entanglement. Suppose Alice and Bob hold systems $A$ and $B$ respectively, and that the three parties perform the following operators. Alice and Bob normal teleport $A$ and $B$ respectively to Charlie. Charlie then port teleports the combined system to Alice who applies the Pauli correction to the $A$ part of each of her $N$ ports. She then port teleports each of these ports to Charlie, who then teleports each of his $N^2$ ports to Ben. Bob then applies his Pauli encryption to the $B$ part of each of his $N^3$ ports. 

Because no communication occurs above, the three parties do not need to wait for each other to perform these operators. Suppose that only Charlie performs the above actions. This defines some tripartite state. Charlie's share consists entirely of the classical information he received from performing port teleportation. Furthermore, there remains no entanglement between the systems of Alice and Bob.

As a resource for the modified NLQC task, choose $V_0'$ and $V_1'$ to be Alice and Bob's part of the tripartite state, and $X_0'$ and $X_1'$ to be copies of Charlie's port teleportation results. Note that there is no entanglement between $V_0'$ and $V_1'$. The protocol proceeds by Alice and Bob performing the actions above with their input systems. Bob then applies the unitary to be implemented on all $N^3$ of his ports, and sends all of the $A$ parts of these to Alice using the one round of communication. Alice and Bob then receive the value of the correct port from the information in $X_0'$ and $X_1'$.

\bibliographystyle{unsrtnat}
\bibliography{biblio}

\end{document}